\documentclass[final,twoside]{IEEEtran}

\normalsize
\usepackage{cite}
\usepackage{algorithmic}
\usepackage{algorithm}
\usepackage{array}
\usepackage{graphicx}
\usepackage{subfigure}
\usepackage{booktabs}
\usepackage{multirow}
\usepackage{amsfonts}
\usepackage{bm}
\usepackage{amsmath}
\usepackage{bbm}
\usepackage{amsthm}
\usepackage{color}
\usepackage{setspace}
\usepackage{subfigure}
\usepackage{setspace}
\newtheorem{proposition}{Proposition}

\begin{document}
	
	\title{Fresh Multiple Access: A Unified Framework Based on Large Models and Mean-Field Approximations}

	\author{Haiming~Hui, Shuqi~Wei, Wei~Chen,~\IEEEmembership{Senior Member,~IEEE}%
		\thanks{Haiming~Hui, Shuqi~Wei and Wei~Chen are with the Department of Electronic Engineering and Beijing National Research Center for Information Science and Technology, Tsinghua University, Beijing 100084, China (e-mail: hhm18@mails.tsinghua.edu.cn; weishuqi77@gmail.com; wchen@tsinghua.edu.cn).}%
		\thanks{This work is supported by the National key research and development program
			under Grant 2018YFA0701601, the National Science Foundation of China under Grant No. 61971264, and the National Natural Science Foundation of China/Research Grants Council Collaborative Research Scheme under Grant No. 62261160390.}
	}
	
	
	\maketitle

	\begin{abstract}
		Information freshness has attracted increasingly attention in the past decade as it plays a critical role in the emerging real-time applications. Age of information (AoI) holds the promise of effectively characterizing the information freshness, hence widely considered as a fundamental performance metric. However, in multiple-device scenarios, most existing works focus on the analysis and optimization of AoI based on queueing systems. The study for a unified approach for general multiple access control scheme in freshness-oriented scenarios remains open. In this paper, we take into consideration the combination of the fundamental freshness metric AoI and multiple access control schemes to achieve efficient cross-layer analysis and optimization in freshness-oriented scenarios, which is referred to as fresh multiple access. To this end, we build a unified framework with a discrete-time tandem queue model for fresh multiple access. The unified framework enables the analysis and optimization for general multiple access protocols in fresh multiple access. To handle the high dimension framework embedded in fresh multiple access, we introduce large model approaches for the Markov chain formulation in AoI oriented scenarios. Two typical AoI-based metric are studied including age of incorrect information (AoII) and peak AoII. Moreover, to address the computational complexity of the large model, we present mean-field approximations which significantly reduces the dimension of the Markov chain model by approximating the integral affect of massive devices in fresh multiple access. 
	\end{abstract}
	
	
	\section{Introduction}
	
	
	The rapid development of 5G and 6G communications has promoted a vast range of real-time applications including industrial internet of things (IIoT), autonomous driving, and unmanned aerial vehicles \cite{roadmap_6G}. In these real-time applications, information freshness plays an increasing role and attracts much research interests. To effectively characterize the information freshness, age of information (AoI) was proposed in \cite{kaul2012real}, which has been widely considered as a fundamental metric. The AoI for single device has been extensively studied in the past decade \cite{kosta2021age, yates2021age}. 
	
	\begin{figure*}[!t]
		\centering
		\includegraphics[width=0.8\linewidth]{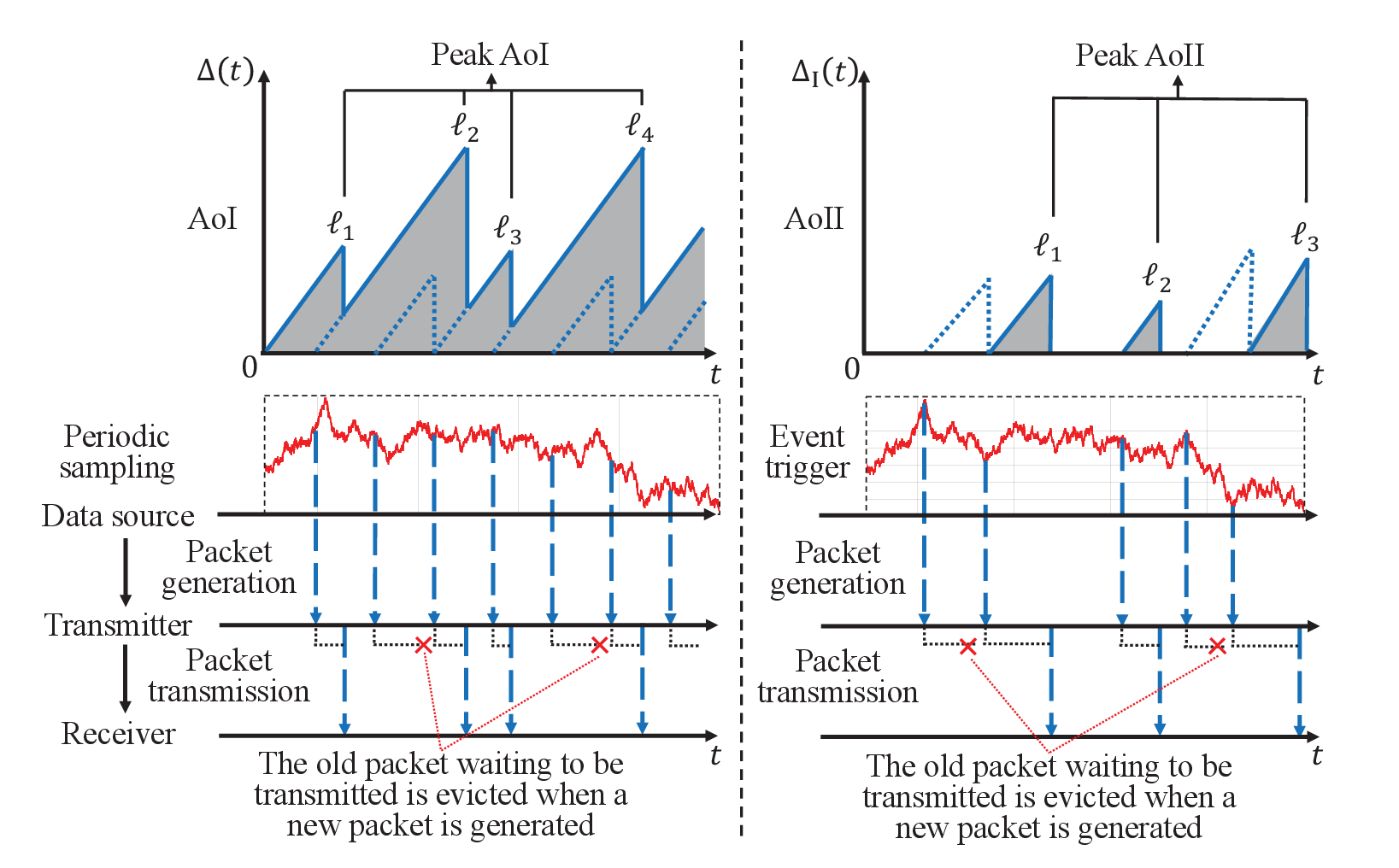}\\
		\caption{The evolution of AoI and AoII over time with periodic sampling and event trigger scheme.}
		\label{evolution_aoii}
	\end{figure*}
	In addition to AoI, extensive connectivity for massive devices are demanded to support the real-time applications. Efficient multiple access control scheme is expected to play a fundamental role in addressing the stringent requirement of connectivity with limited communication resources. However, conventional multiple access techniques can only be applied with a relatively small number of devices \cite{conventional_survey_2, conventional_access}. To handle the access control of massive devices, random access schemes were developed as a distributed control protocol, which allows dynamic resource allocation for devices in a contention manner, hence attracting much research interest in massive access control \cite{ephremides1989control, massive_access_future}. The earliest random access algorithm called Aloha was first developed in \cite{aloha_original}. The Aloha algorithm is readily to implement because of its simple rules. Extensive analysis for Aloha algorithm were studied concerning its throughput and stability regions \cite{luo1999stability, luo2006throughput}. Furthermore, the tree splitting algorithm \cite{splitting_1} and first-come-first-serve (FCFS) splitting algorithm \cite{splitting_2} with more sophisticated collision resolution rules were proposed to achieve higher throughput and lower delay than the simple Aloha algorithm. 
	
	In conventional metrics for information freshness, the state change or update of the monitored process itself is not taken into consideration. To overcome this limitation, a novel metric for the age of information is defined to be the time difference between the state change and remote update, for both monitor processes with alarms \cite{addterm} and counting processes \cite{wangmeng}.To improve freshness in multi-user systems, AoI was considered in conventional multiple access schemes. However, the analysis and optimization of multiple access schemes in AoI oriented scenarios are quite challenging since the system is affected by the collective influence of multiple users. Most existing works focused on the characteristic of queueing systems in the analysis or optimization of AoI. Two typical sampling schemes are shown in Fig. \ref{evolution_aoii}. With periodic sampling or Poisson arrivals of data packets, old data packets are evicted to improve freshness. In this case, the scheduling policies based on queueing systems were studied for the analysis and optimization of average AoI \cite{kam2018age}. The NOMA scheme and orthogonal multiple access (OMA) schemes are studied to reduce the average AoI \cite{maatouk2019minimizing}. The tradeoff between the average AoI and drop rate for two-user multi-access scenario was revealed in \cite{fountoulakis2022information}. Furthermore, to address the access control for massive devices, random access based access protocols were applied to reduce the signaling overhead. The collision resolution algorithm for slotted Aloha scheme \cite{yates2017status} and pure Aloha scheme \cite{yates2020age} were further optimized to reduce the average AoI. Threshold-based access policy were investigated to reduce the average AoI by dynamically adjusting the access frequency \cite{chen2022age}. When the instability embedded in fading channels is considered, a rate-adaptive transmission scheme was studied to minimize the average AoI under an average power constraint \cite{wang2021adaptive}. Since information content of data packets can be further exploited to achieve fresh applications, event-trigger scheme was considered to capture more informative data packets from the data source. All data packets are transmitted without eviction. In this case, timing side information was utilized in compression for the time stamp of data packets to evaluate AoI \cite{yu2023real}. Compression distortion and AoI are jointly optimized in real-time monitoring to minimize the reconstruction distortion \cite{hui2022real}. The AoI was reduced under the sampling scheme based on the scheduling policy \cite{kadota2019scheduling}.
	 A new AoI-based metric called age of incorrect information (AoII) was proposed as a new performance metric through event trigger scheme \cite{maatouk2020age}. Through event trigger schemes, AoII was shown to be able to capture more meaningfully the purpose of data, thus attracting much attention in semantic communications \cite{maatouk2023age}. The scheduling policy for multiple sensors in slotted Aloha systems was studied to minimize the AoII \cite{nayak2023decentralized}.

	In addition, peak AoI plays an important role in freshness oriented multiple access since the peak AoI characterizes the staleness of the transmitted information \cite{costa2014age}, as illustrated in Fig. \ref{evolution_aoii}. It is shown that peak AoII depends on the access delay \cite{kosta2020cost}. Multiple access schemes can be further optimized to reduce the peak AoII or access delay. A cross-layer approach with NOMA was studied to minimize the average delay \cite{noma_delay_sensitive}, which is further generalized into scenarios with arbitrary packet arrivals and adaptive transmission \cite{liu2022joint}. A mean-field approximation approach was adopted aiming at the analysis of delay-optimal scheduling \cite{mean_field}. Slotted uncoordinated random access schemes were developed to serve a massive number of devices with quality-of-service requirements guaranteed \cite{connectionless_access}. Polling schemes allow the central device to ask each device in sequence to conduct data transmission \cite{polling_book}. A device only consumes a short timeslot if the device has no data packets to transmit. Performance analysis of polling schemes were studied based on queueing models \cite{polling_analysis_queuing}. The scheduling policy under polling scheme were studied to minimize the average AoI with stochastic packet generation model \cite{kosta2019age}. Due to the high complexity in modeling the multiple access scheme for massive devices in AoI oriented scenarios, the analysis and optimization for multiple access scheme with more devices under various AoI-based performance metrics remains open. 
	
	Furthermore, reservation-based random access schemes hold the promise of addressing massive access control with limited resources. Particularly, the throughput of the Aloha and FCFS algorithm are around 0.368 and 0.487. The upper bound of the throughput is 0.568 \cite{upperbound_throughput}, which means considerable resources are inevitably wasted from the collision. While the maximum throughput is limited due to the inevitable collision in the random access process \cite{throughput_survey}, access protocols based on reservation techniques were investigated to approach the throughput of one \cite{wieselthier1988distributed}. In reservation-based schemes, each device needs to make a reservation prior to its data transmission. A basic reservation-based multiple access scheme through satellite network was proposed in \cite{AE_satellite}, in which each minislot is allocated to a fixed node to make reservation. To extend the connectivity, random access techniques are adopted in the reservation procedure. Thus, all nodes were allowed to make reservation in any given minislot in a contention-based manner \cite{capacity_reservation}. The reservation-based random access\footnote{The reservation-based random access is also referred to as the connection-based random access in the literature, since the reservation process can be regarded as setting up the connection.} scheme may significantly improve the throughput by sending short reservation signals to reserve resources for collision-free packet transmission \cite{reservation_connection_based}.

	The main contents of this paper are illustrated in Fig. \ref{content_structure}. In this paper, we build a unified framework to provide a general approach to analysis, optimization, and comparison of all these multiple access schemes in freshness-oriented scenarios, which is referred to as fresh multiple access. The general multiple access scheme is characterized by three consecutive stages. First, in the access trigger stage, new reservation is triggered at a user based on the threshold of the local buffer. Second, in the reservation stage, the user sends reservation signals in a contention or polling based manner. Third, in the transmission stage, the user transmits the data packets based on its reservation in a contention-free manner. Reservation-free multiple access schemes are also unified in this framework by considering zero service time for the transmission stage. In this case, the data packets to be transmitted is regarded as the reservation signal. Throughout the access procedure, a tandem queue structure including a virtual reservation queue and a virtual transmission queue are built in the protocol layer to characterize the users in each stages.
	
	Based on the unified framework, general multiple access scheme is modeled with Markov chain formulation. However, when the freshness metric is incorporated in the Markov chain model, the dimension of the Markov chain model for fresh multiple access becomes prohibitively high. To this end, we present large model based approaches for analysis and optimization of fresh multiple access. To formulate the large model Markov chain for fresh multiple access, we obtain the sparse transition probabilities for each state based on the characterization of multiple access protocols and arbitrary packet generation process. Thus, the whole transition matrix of the large model Markov chain can be obtained. Further analysis and optimization can be conducted based on the large model Markov chain. Note that arbitrary mechanisms can be adopted for the reservation stage, hence enabling further optimization towards the multiple access scheme. We focus on typical reservation schemes including polling, slotted Aloha, and tree splitting algorithm in the reservation stage to demonstrate the unified framework in this paper. Moreover, to address the high computational complexity of the large model for massive devices oriented scenarios, we present mean-field approximations for the performance analysis of fresh multiple access. Through mean-field  approximation, the integral affect of massive devices in fresh multiple access are approximated as an environment affect. Thus the dimension of Markov chain is significantly reduced to comprise only the local state of a single user. 
	
 	 Based on the Markov chain model, we are able to analyze the TD, FD, and XD multiplexing schemes for either AoII or peak AoII. Specifically, for AoII oriented scenarios, we formulate large Markov model based on individual states of all devices in the system. For peak-AoII oriented scenarios, we only need the total number of data packets to obtain the average peak AoII. Thus, we formulate reduced-dimensional Markov model based on integral states of the whole system.
	The main contributions in this paper are listed as follows.
	\begin{itemize}
		\item[1)] We study two typical fresh multiple access scenarios oriented by AoII and peak AoII based on the unified framework. 
		\item[2)]  We formulate the large model Markov chain to characterize various multiple access protocols and multiplexing schemes.	Three multiplexing schemes of reservation signals and data packets are considered, including multiplexing in the time domain (TD) \cite{conference_TD}, frequency domain (FD), and dynamic bandwidth allocation scheme (XD).
		\item [3)]In massive devices oriented scenarios, we present mean-field approximations to analyze the AoII and peak AoII, which simplifies the Markov chain formulation and reduces the computational complexity of large model. Thus, we formulate a small model Markov chain based on the local state of a single user with mean-field approximation. Based on the Markov chain, we can compute the steady-state probabilities which leads to the AoII or peak AoII performance metric.
		\item [4)]  Markov decision process (MDP) is applied to optimize the dynamic scheduling policy for the XD scheme. We present extensive numerical results to demonstrate the analysis for AoII and peak AoII with various access protocols based on the three kinds of Markov chain formulation. 
	\end{itemize}
	
	\begin{figure}[!t]
		\centering
		\includegraphics[width=1\linewidth]{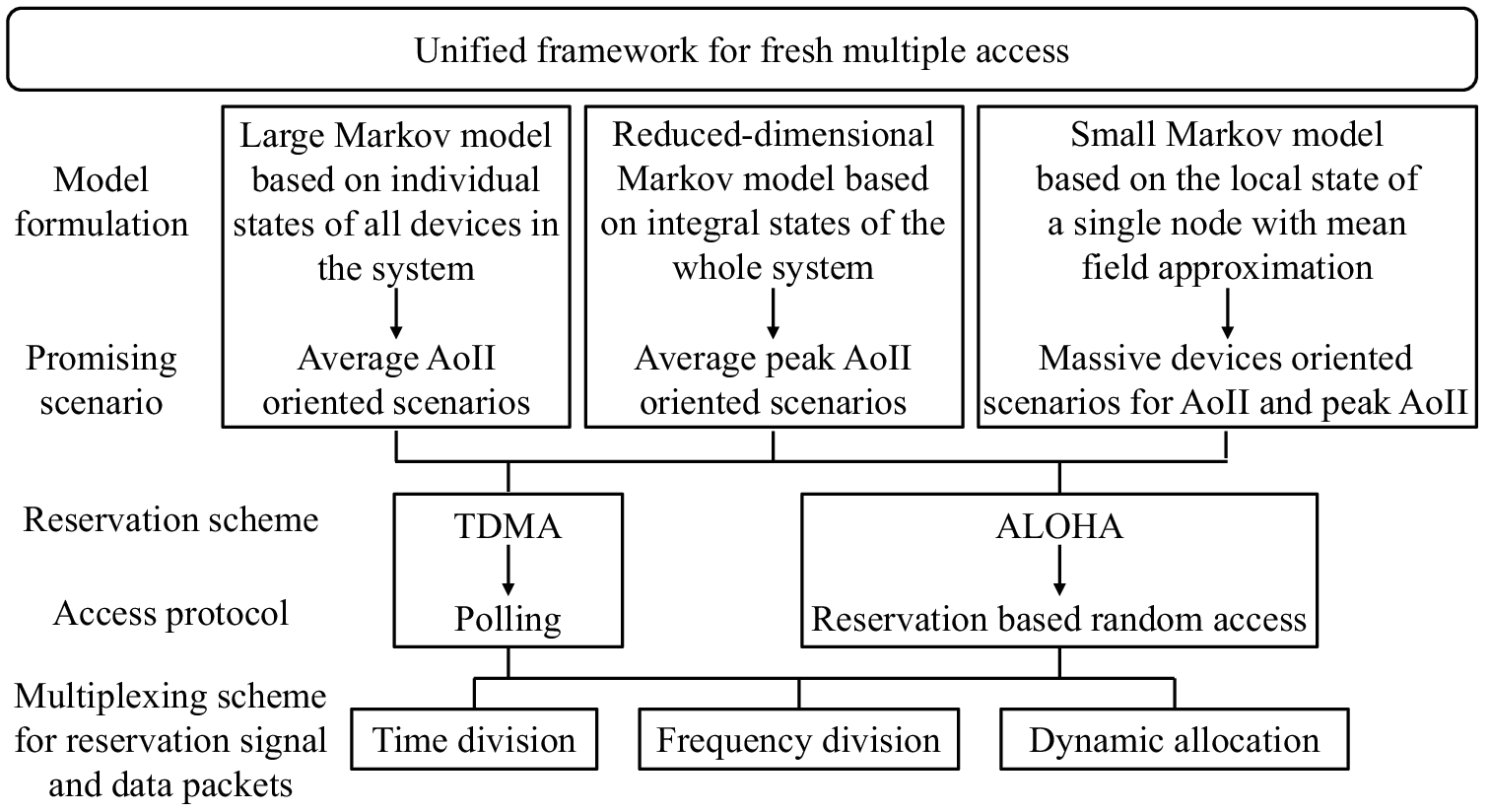}\\
		\caption{The illustration of the main contents of this paper based on the unified framework for multiple access.}
		\label{content_structure}
	\end{figure}

	The rest of this paper is organized as follows. The system model and the unified framework are described in Section II. The Markov chain formulation for AoII oriented scenarios and peak-AoII oriented scenarios are presented in Sections III and IV, respectively. Markov chain formulation with mean-field  approximation for massive devices oriented scenario is presented in Section V. Numerical results are presented in Section VI. Finally, we present our conclusions in Section VII.
	
	\section{Related Work}
	With reservation-based random access schemes, the freshness was further improved in scenarios with massive devices. Delay and stability of the reservation-based slotted Aloha scheme were analyzed in \cite{connection_slotted_aloha} based on the queueing model. Two typical reservation-based random access schemes are applied in the 802.11 protocol and the long term evolution (LTE) cellular networks. The primary access protocol of 802.11 is based on the carrier sense multiple access with collision avoidance (CSMA/CA) scheme, which defines request-to-send/clear-to-send (RTS/CTS) signals as the reservation signal. An analytical model to compute the throughput of CSMA schemes was presented by Bianchi in \cite{Bianchi_CSMA_80211}. More detailed analysis for particular p-persistent and non-persistent CSMA were studied in \cite{CSMA_p_persistent, CSMA_non_persistent}. The access delay of CSMA with unsaturated networks was investigated in \cite{queue_delay_80211} based on the queue modeling for nodes in the network. The scheduling policy based on the CSMA scheme can be optimized to reduce AoI in scenarios with power constraints \cite{bedewy2021low}. The average AoI with CSMA scheme was optimized in both the sampling scenarios and stochastic arrivals scenarios \cite{maatouk2020ageCSMA}. Moreover, in the LTE networks, a physical random access channel, which appears periodically in time frames, is used to transmit the reservation signal referred to as preamble. The throughput and access delay of machine-to-machine (M2M) communications in LTE networks were optimized by tuning parameters of the inherent Aloha scheme \cite{Dai_modeling_throughput, Dai_access_delay}. The resource consumption and throughput of random access were studied to obtain Pareto-optimal configuration \cite{random_access_consumption}. The double-queue model presented in \cite{Dai_modeling_throughput} does not consider the resource consumption of the packet transmission in its second queue. However, by taking into consideration the both the reservation and the packet transmission, we can obtain a more comprehensive model with unified analysis and optimization framework for freshness-oriented multiple access.	
	
	\begin{figure}[!t]
		\centering
		\includegraphics[width=0.7\linewidth]{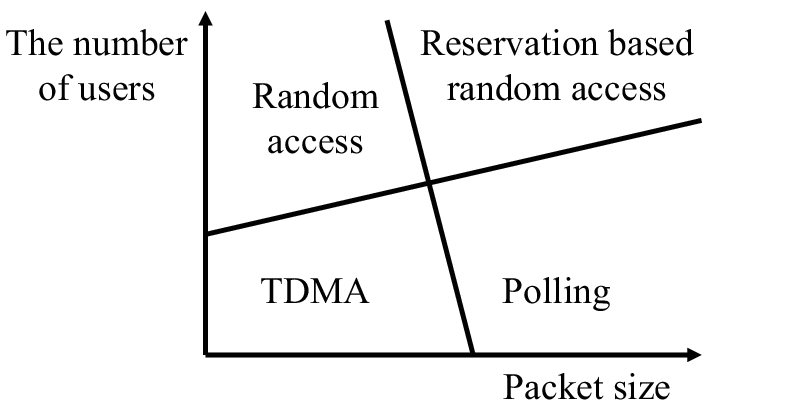}\\
		\caption{The suitable multiple access schemes under different cases of the number of devices and packet size.}
		\label{fig_four_phase}
	\end{figure}
	
	For the various widely used multiple access techniques, they are categorized into four types from the perspective of contention and reservation, as shown in Fig. \ref{fig_four_phase}, including time division multiple access (TDMA), polling scheme, random access (such as Aloha), and reservation based random access (such as CSMA/CA and LTE networks). It is shown that each scheme is shown to be superior in particular scenarios of different connectivity and packet size. Specifically, on one hand, TDMA and polling are contention-free schemes in which all users are allocated resources by the BS to access the channel in a contention-free manner. Random access and reservation based random access are contention-based schemes in which the BS does not guarantee resource allocation for each device. All devices attempt to access the channel in a contention-based manner. The contention-based manner is more suitable for scenarios with massive devices with sporadic packet generations. Thus, as the number of devices increase, the contention-based manner is shown to outperforms contention-free manner. On the other hand, TDMA and random access are reservation-free schemes, which means that the data packets are transmitted directly when the device access the channel. Polling and reservation based random access are reservation based schemes, which means that the device should send a reservation signal prior to its packet transmission. The reservation signal is transmitted based on TDMA or random access schemes in Polling or reservation based random access, respectively. The reservation mechanism alleviates the waste of resources from idle and collision. Thus, the reservation-based scheme is shown to outperform the reservation-free scheme in scenarios with large packet size.
	\section{System Model}
	
	We consider a multi-user system, of which our goal is to optimize the information freshness. To characterize the information freshness of multiple users, we consider the AoI metric. There are $M$ users and a receiver in the system. Each user is regarded as a transmitter node, which is equipped with an infinite buffer to store data packets to be transmitted. Data packets of status updates are generated at each node and transmitted to the receiver. With periodic sampling schemes, the interval of status updates at each node is the same. Thus, the TDMA scheme can be applied, in which each node generates the data packet at the beginning of the timeslot allocated to the node for transmission. 
	
	However, periodic sampling may not capture the fresh informative updates efficiently. To address that, we consider incorrect information with event trigger scheme for the multiple access network model, as shown in Fig. \ref{evolution_aoii}. Specifically, with event trigger scheme, each node only generates a data packet of status update when the status changes, which leads to the AoII metric. The node does not transmit any data packets if the status remains the same, hence reducing the resource consumption. To study access control towards incorrect information, we build a unified framework including two virtual queues referred to as the reservation queue and the transmission queue, as shown in Fig. \ref{system_model}. The unified framework can be adopted for general multiple access schemes. We assume that the packet generation at each node is a Poisson process with the same arrival rate $\bar{\lambda} = \frac{\lambda}{M}$. The packet arrival for the whole network is also Poisson with rate $\lambda$. Specifically, the probability that there are $i$ packets arriving at all nodes in a unit of time is given by
	\begin{align}
		a_{i} = \frac{(\lambda)^i}{i!}e^{\lambda}.
	\end{align}

	\begin{figure*}[!t]
		\centering
		\includegraphics[width=1\linewidth]{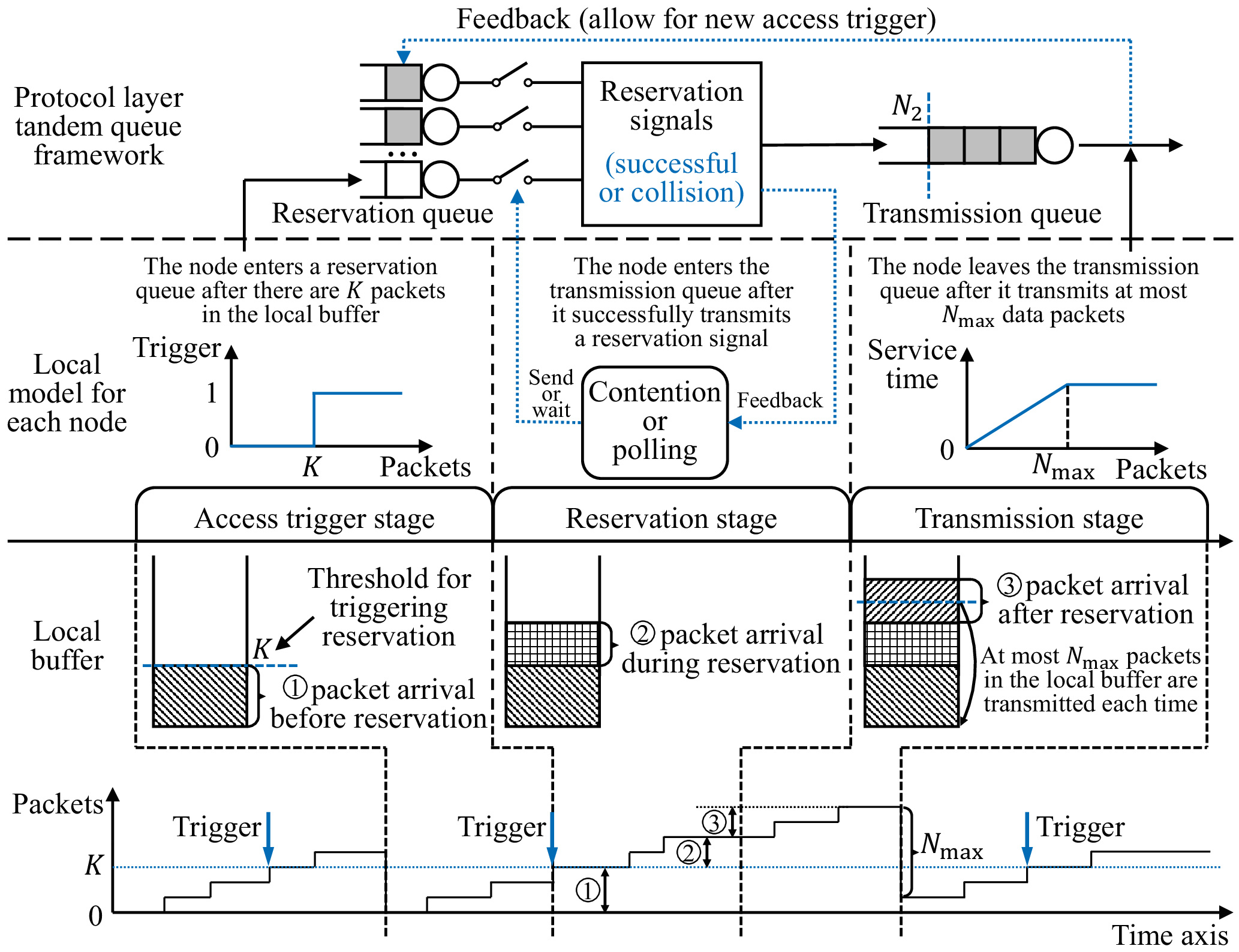}\\
		\caption{The unified framework for multiple access schemes and the tandem queue.}
		\label{system_model}
	\end{figure*}
	
	The local access model of the unified framework for each node is divided into three stages. First in the access trigger stage, the node triggers a new access attempt when there are $K$ data packets in the local buffer. Thus, the threshold for triggering reservation is defined by $K$. After triggering a new access attempt, the node enters the reservation queue to begin the reservation. By sending only one reservation signal for $K$ data packets, the number of reservation signals to be transmitted is reduced, which can be regarded as reducing the arrival rate for the reservation queue. Next in the reservation stage, the node sends the reservation signal in a contention-based manner based on the collision resolution algorithm or in a contention-free manner. The collision resolution algorithms are also known as random access algorithms. The node enters the transmission queue after it successfully transmits a reservation signal. Third in the transmission stage, the node transmits data packets sequentially in the transmission queue. When the node comes to the head of the transmission queue, the node transmits at most $N_{\mathrm{max}}$ data packets, which is referred to as the transmission constraint, in the local buffer. Note that $N_{\mathrm{max}}$ can be greater than $K$ since new data packets can arrive during the reservation and transmission stage. If there are more than $N_{\mathrm{max}}$ data packets in the local buffer, the node may trigger a new reservation even when it still stays in the transmission queue. In the access procedure, each node enters the reservation queue and then the transmission queue. Since the departure process of the reservation queue is the same as the arrival process of the transmission queue, the two virtual queues constitute a tandem queue model. The reservation-based access protocol maintains a finite length of the tandem queue. Thus, we assume that the maximum length of the reservation queue and the transmission queue is given by $N_1$ and $N_2$, respectively. 
	
	For the reservation stage, each node needs to successfully send a reservation signal prior to its packet transmission. An ACK signal is sent by the recerver after a reservation signal is successfully transmitted through a broadcast channel to inform all nodes. Thus, all nodes are aware of the current state of the reservation, which ensures the packet transmission based on reservation does not collide with each other. Based on the reservation of nodes, all data packets are transmitted sequentially in a contention-free manner. The sending of reservation signals is conducted in a contention-based manner of contention-free manner. Specifically, with the contention-based manner, each node decides to send the reservation signal based on the collision resolution algorithm such as Aloha. A reservation signal is successfully transmitted by a node only when no other nodes send reservation signals in the same timeslot. In other words, if more than one node sends the reservation signal in the same timeslot, then a collision occurs in that timeslot while no reservation is made in that timeslot. It addresses the sporadic packet arrivals and the massive number nodes. With the contention-free manner, the receiver asks each node in a cyclic order whether the node has data packets to transmit. If a node has triggered an access and entered the reservation queue, then the node should answer the receiver when the receiver asks the node, which is regarded as successfully transmission of the reservation signal. After that, the node enters the transmission queue. Therefore, each node in the reservation stage waits for the receiver to ask instead of trying to send the reservation signal according to collision resolution algorithms, hence alleviating the collision among nodes. However, with the contention-free scheme, the node must wait for the receiver to ask all nodes one by one, which brings extra AoII in the reservation stage especially in scenarios with a massive number of nodes. 
	
	The reservation signal and data packets are multiplexed in the time domain or frequency domain for transmission. We consider that the bandwidth for the reservation signal and the data packets are $w_1$ and $w_2$, respectively. Since the transmission rate is proportional to the bandwidth, we simply denote the transmission rate as $w_1$ and $w_2$. The total transmission rate is constrained by $w = w_1 + w_2$. Without loss of generality, we consider $w = 1$ throughout this paper. The bandwidth $w_1$ and $w_2$ can be fixed or adjusted dynamically to implement different access control protocols. We assume the size of each data packet is $c$ while the size of a reservation signal is one. The time for sending the ACK signal by the receiver is assumed to be negligible. The time for transmitting a reservation signal is given by
	\begin{align}
		T_1 = \frac{1}{w_1},
	\end{align}
	which is inversely proportional to the bandwidth $w_1$. The time for transmitting each data packet is given by
	\begin{align}
		T_2 = \frac{c}{w_2}.
	\end{align}
	We consider the AoII metric to represent the performance of the multiple access network. To improve the freshness of the status updates, each node only stores the freshest data packet in the local buffer. The AoII is defined as the time elapse since a data packet is generated in the local buffer. When a new data packet is generated, the old one is evicted while AoII continues increasing with time. Thus, the AoII represents the freshness of information, which is defined as
	\begin{align}
		\Delta_{I}(t) &= \begin{cases}
			0, & \text{if no new data packet is generated}\\ &\text{after the last packet transmission,} \\
			t - U(t), &\text{otherwise},
		\end{cases} 
	\end{align}
	where $U(t)$ is the generation time of the first new data packet after the last packet transmission by the node. To analyze the average AoII metric, it requires large-model approaches since the Markov chain model can be high-dimensional. 
	
	\begin{proposition} 
		In multi-access communications where each user exclusively retains its own freshest packet, its average peak AoII is equal to the average waiting time of its successfully delivered packets.		
		\label{lem-1}
	\end{proposition}
	\begin{proof}
			In the scenario where each user exclusively retains the freshest packet in its buffer and computes the AoII independently, the AoII is equal to the duration a packet remains within the multi-access communications. Furthermore, the peak AoII corresponds to the total time a packet resides within the system. This alignment arises from the fact that each instance of a peak AoII occurs precisely when a packet is received. Consequently, the average peak AoII is equal to the average time a packet spends in multi-access communications.
	\end{proof}
	Furthermore, we consider the average peak AoII metric to reduce the dimension of the modeling towards average peak AoII. Specifically, let $\bar{L}$ denote the average number of data packets in the tandem queue. The peak AoII of a data packet is defined as the maximum value of AoII achieved before the reception of a data packet. The average peak AoII in the second and third stages is obtained through the Little's Law, given by
	\begin{align}
		\label{eq_average_delay}
		\ell = \frac{\bar{L}}{\lambda}.
	\end{align}
	Thus, we only need to track the state of the whole system in terms of the total queue length, instead of the individual state of each transmitter node, which simplifies the Markov chain formulation. A data packet may experience an extra delay in the access trigger stage when $K > 1$, which increases the peak AoII. The average time of waiting for a data packet arrival at a node is given by $\frac{M}{\lambda}$. Thus, the peak AoII in the first stage is given by
	\begin{align}
		\ell_0 &= \frac{M}{\lambda K}l\sum_{i=0}^{K-1}i \nonumber \\&=\frac{M(K-1)}{2\lambda}.
	\end{align}
	Increasing $K$ for the access trigger stage can reduce the signaling overhead while the peak AoII increases. In addition, the peak AoII $\ell_0$ increases linearly with the number of nodes in the network $M$. As the number of data packets transmitted related to each reservation signal increases, the signaling overhead for the reservation signal can approach zero. Thus, to attain finite peak AoII, the total arrival rate should be less than the maximal transmission rate of the transmission queue. Thus, the total arrival rate satisfies that $\lambda < \frac{w}{c}$.

	\section{Large Markov Model for AoII Oriented Multiple Access}
	
	In this section, we formulate a large model Markov chain based on the unified framework towards the AoII metric. The local buffer of each node only stores the freshest data packet. We represent the system state based on the state of each node. The state of node $i$ is $(q_i, e_i)$ in which $q_i$ represents the AoII at the node and $e_i$ represents the reservation state of the node. Let $s \in \{1, \ldots, M\}$ denote state of the reservation procedure. When a node begins transmission of data packets, the remaining length of data packets to be transmitted is denoted by $t_d \in \{0, \ldots, c\}$. For $t_d=0$, it means that nodes in the reservation queue are trying to send reservation signal in the current timeslot based on their reservation state. Thus, the system state is denoted by $S = (s,e_1, \ldots, e_M, q_1,\ldots,q_M,t_d)$. 
	
	Both reservation signals and data packets consumes bandwidth resources for transmission. There are three typical multiplexing schemes for reservation signals and data packets including time-division (TD), frequency-division (FD), and dynamic bandwidth allocation scheme (XD). In this section, to reduce the average AoII, we focus on the XD scheme. Specifically, after a node successfully transmits a reservation signal, the node begins transmitting the data packet in the next timeslot while the transmission lasts for $c$ timeslots. Other nodes cannot send the reservation signals until the current node finishes its transmission stage. 
	
	In each timeslot, the state transition is divided into two steps. First, each node send the reservation signal according to its reservation state and the state of the current reservation procedure. Second, new data packets may arrive at each node. In the first step, if $t_d > 0$, then a node is in the transmission queue and transmitting the data packet. Thus, the state $t_d$ transits to $t_d - 1$ while other states do not change. If $t_d = 0$, then nodes in the reservation queue send reservation signals based on the polling or contention mechanism. The reservation state of each node $e_i$ and the state of the reservation procedure $s$ changes accordingly. After the first step, the state transits from $S = (s, e_1, \ldots, e_M, q_1, \ldots, q_M, t_d)$ to intermediate state $S' = (s', e_1', \ldots, e_M', q_1', \ldots, q_M', t_d')$. 
	
	In the second step, the transition probability is obtained based on the packet arrival distributions. Let $\bar{a}$ denote the probability that there are at least one data packets arriving at a node in a timeslot, given by
	\begin{align}
		\bar{a} = 1 - e^{-\bar{\lambda}}.
	\end{align}
	The state $q_i'$ is the age of the data packets at node $i$. Thus, when $q_i' > 0$, then the age increases by one. When $q_i' = 0$, it transits to $q_i'' = 1$ with probability $\bar{a}$. The state $s', e_1', \ldots, e_M'$ and $t_d'$ do not change in the second step. Consider the constraint of each node's AoII is denoted by $N$, which is assumed to be finite for the purpose of formulating finite-dimension Markov chain.\footnote{By setting the constraint $N$ large enough, the probability that the age exceeds $N$ can approach zero. Thus, this constraint may not affect the performance of AoII metric.} Thus the transition probability of state $q_i'$ is given by
	\begin{align}
		\label{eq_transition_aoii_1}
		p'_{i,q_i',q_i''} = \begin{cases}
			\bar{a}, & q_i'' = 1, q_i' = 0, \\
			1 - \bar{a}, & q_i'' = q_i' = 0, \\
			1, & q_i'' = \min(q_i' + 1, N), q_i' > 0.
		\end{cases}
	\end{align}
	Since packet arrival for all nodes are independent, the transition probability from $S' = (s', e_1', \ldots, e_M', q_1', \ldots, q_M', t_d')$ to $S'' = (s'', e_1'', \ldots, e_M'', q_1'', \ldots, q_M'', t_d'')$ is given by
	\begin{align}
		\label{eq_polling_transition_step2}
		p'_{S', S''} = \prod_{i=1}^{M}p'_{i,q_i',q_i''}
	\end{align}
	\subsection{Polling Scheme}
	We first consider the polling scheme in which the reservation signals are transmitted in a collision-free manner. Since the receiver asks each node in turn to send the reservation signal, the reservation state of each node is represented by waiting or being asked in the current timeslot. The state of the reservation procedure $s \in \{1, \ldots, M\}$ denotes the index of node in the polling procedure that the receiver currently asks. For simplicity, we omit the reservation state $e_i$ of each node. Thus, the system state is denoted by $S = (s,q_1,\ldots,q_M,t_d)$. The entire state space is denoted by $\mathcal{S}$.
	
	In the first step, when $t_d = 0$, the receiver asks the next node $s'$ (of index $s+1$ if $s < M$ or index $0$ if $s = M$). Then node $s'$ begins transmitting data packets of length $c q_{s'}$. Note that if $q_{s'} = 0$, then the node actually has nothing to transmit. As a result, the state transits from $S = (s, q_1, \ldots, q_M, t_d)$ to intermediate state $S' = (s', q_1', \ldots, q_M', t_d')$ that is determined by $S$. Specifically, state $s'$ and $d'$ is given by
	\begin{align}
		s' &= \begin{cases}
			s + 1 , & t_d = 0, s < M , \\
			1, & t_d = 0, s = M, \\
			s, & t_d > 0,
		\end{cases} \label{eq_relation_d1}\\
		t_d' &= \begin{cases}
			q_{s'}c, & t_d = 0 \\
			t_d - 1, & t_d > 0.
		\end{cases}
	\end{align}
	For $i = 1, \ldots, M$, the state $q_i'$ is given by
	\begin{align}
		\label{eq_transition_polling_q}
		q'_{i} &= \begin{cases}
			0, & t_d = 1, i = s, \\
			q_i, & \mathrm{otherwise}.
		\end{cases} 
	\end{align}

	Based on Eqs. (\ref{eq_transition_aoii_1})-(\ref{eq_transition_polling_q}), we can derive the transition probability matrix of the Markov chain model. The transition probability matrix is denoted by $P$, of which the element $p_{S,S''}$ is the transition probability from $S$ at the beginning of the current timeslot to state $S''$ at the beginning of the next timeslot.  Specifically, the element $p_{S,S''}$ for each pair $S, S''$ is derived by first obtaining the intermediate state $S'$ from $S$. Then the probability is given by
	\begin{align}
		p_{S,S''} = p'_{S',S''}.
	\end{align}

	The transition matrix $P$ is of order $r_1 = (N+1)^M M (c+1)$. The transition matrix is sparse since each state can only transits to a few states. Thus, we can obtain the transition probabilities starting from each state to reduce the complexity. The methods of obtaining the transition matrix is shown in Algorithm \ref{algorithm_transition_matrix}. Since any multiple access protocols can be regarded as a formal language \cite{hopcroft_book}, the algorithm to generate the transition matrix can be adopted for any multiple access protocols. We can also obtain the large but sparse transition matrix for other multiple access protocols through such algorithms accordingly. 
	
	\begin{algorithm}[!t]
		\caption{Algorithm to obtain transition matrix}
		\begin{algorithmic}[1]
			\label{algorithm_transition_matrix}
			\STATE Initialize $p_{S,S''} \leftarrow 0$ for all $S,S'' \in \mathcal{S}$
			\FORALL{$S=(s,q_1,\ldots,q_M,t_d)\in\mathcal{S}$}
			\STATE $S'=(s',q_1',\ldots,q_M',t_d')\leftarrow S$
			\IF{$t_d = 0$}
			\IF{$s < M$}
			\STATE $s' \leftarrow s + 1$
			\ELSE
			\STATE $s' \leftarrow 1$
			\ENDIF
			\STATE $t_d' \leftarrow c \mathbbm{1}_{\{q_{s+1} > 0\}}$
			\ELSE
			\STATE $t_d' \leftarrow t_d - 1$
			\ENDIF
			\IF{$t_d = 1$}
			\STATE $q_{s+1}' = 0$
			\ENDIF
			\FORALL{$(v_1, \ldots, v_M) \in \{0,1\}^{M}$}
			\STATE $S''=(s'', q_1'', \ldots, q_M'',t_d'') \leftarrow S'$
			\STATE $x = 1$
			\FOR{$k = 1,\ldots,M$}
			\IF{$q_{k+1}'=0$}
			\STATE $q_{k+1}'' \leftarrow v_k$
			\IF{$v_k = 1$}
			\STATE $x \leftarrow x \bar{a}$
			\ELSE
			\STATE $x \leftarrow x (1 - \bar{a})$
			\ENDIF
			\ELSE
			\STATE $q_{k+1}'' \leftarrow \min\{q_{k+1}'+1, N\}$
			\IF{$v_k = 1$}
			\STATE $x = 0$
			\ENDIF
			\ENDIF
			\ENDFOR
			\STATE $p_{S,S''} \leftarrow x$
			\ENDFOR
			\ENDFOR 
		\end{algorithmic}
	\end{algorithm}
	
	Then we can compute the steady-state probabilities of all states. Let a row vector $\bm{\pi}$ denote the steady-state probabilities, which satisfy
	\begin{align}
		\begin{cases}
			\bm{\pi} P = \bm{\pi}, \\
			\bm{\pi} \bm{1}_{r_1} = 1. 
		\end{cases}
	\end{align}
	
	Thus, we have $\bm{\pi} \left(P - I_{r_1} + \bm{1}_{r_1\times r_1}\right) = \bm{1}_{r_1}^{\mathrm{T}}$. The steady-state probabilities are given by
	\begin{align}
		\label{eq_steady_state_polling}
		\bm{\pi} = \bm{1}_{r_1}^{\mathrm{T}} \left(P - I_{r_1} + \bm{1}_{r_1\times r_1}\right)^{-1}. 
	\end{align}
	To address the high dimension embedded in the state space and transition matrix, large model based approaches are applied to solve the large sparse chain \cite{large_sparse_chain_1, large_sparse_chain_2}.
	
	The steady-state probability for a given state $S=(s,q_1,\ldots,q_M,t_d)$ is denoted by $\pi_{s,q_1,\ldots,q_M,t_d}$. The average AoII of all nodes are given by
	\begin{align}
		\label{eq_delta_polling}
		\bar{\Delta} = \sum_{q_1=0}^{N} \cdots \sum_{q_M=0}^{N} \sum_{i=1}^{M}\frac{q_i}{M} \sum_{t_d=0}^{c} \sum_{s=1}^{M} \pi_{s,q_1,\ldots,q_M,t_d}.
	\end{align}

	\subsection{Random Access Scheme}
	
	We next consider the random access scheme to allow each node to send reservation signals in a contention-based manner. The contention during the reservation stage is handled by the collision resolution algorithm. Thus, we first present a general characterization for arbitrary collision resolution algorithms, which are represented by Markov chains. Let $\mathcal{S}_{\mathrm{R}}$ denote the state space of the reservation procedure and $r_0 = \left|\mathcal{S}_{\mathrm{R}}\right|$ denote the size of the state space. The transition of the Markov chain represents the process of contention among reservation signals in the reservation queue during each timeslot, while the reservation state of each node $s_i$ changes accordingly. On one hand, if no reservation signal is successfully transmitted in a timeslot, then the transition probability from state $s_i\in\mathcal{S}_{\mathrm{R}}$ to state $s_j\in\mathcal{S}_{\mathrm{R}}$ is given by $X_{0,ij}$. We can obtain an $r_0\times r_0$ transition matrix $X_0$ with elements $X_{0,ij}$ for $1 \leq i \leq r_0$ and $1 \leq j \leq r_0$, given by
	\begin{align}
		X_0 = \left[\begin{array}{ccc}
			X_{0, 11} & \cdots & X_{0, 1r_0} \\
			\vdots & \ddots & \vdots \\
			X_{0, r_01} & \cdots & X_{0, r_0r_0}
		\end{array}\right].
	\end{align}
	On the other hand, if a reservation signal is successfully transmitted in a timeslot, then the transition probability from state $s_i \in \mathcal{S}_{\mathrm{R}}$ to state $s_j \in \mathcal{S}_{\mathrm{R}}$ is given by $X_{1,ij}$, which constitute an $r_0\times r_0$ transition matrix $X_1$ given by
	\begin{align}
		X_1 = \left[\begin{array}{ccc}
			X_{1, 11} & \cdots & X_{1, 1r_0} \\
			\vdots & \ddots & \vdots \\
			X_{1, r_01} & \cdots & X_{1, r_0r_0}
		\end{array}\right].
	\end{align}
	At the beginning of each timeslot, a new collision resolution process (CRP) may begin. The state for a new CRP depends on the length of the reservation queue. When the current state is $s_i$ and there are $n$ nodes in the reservation queue, the transition probability from state $s_i$ to state $s_j$ is denoted by $Y_{n,ij}$. We can obtain $r_0\times r_0$ transition matrices $Y_n$ for $n = 0, 1, \ldots$ with elements $Y_{n,ij}$ for $1 \leq i \leq r_0$ and $1 \leq j \leq r_0$, given by
	\begin{align}
		Y_n = \left[\begin{array}{ccc}
			Y_{n, 11} & \cdots & Y_{n, 1r_0} \\
			\vdots & \ddots & \vdots \\
			Y_{n, r_01} & \cdots & Y_{n, r_0r_0}
		\end{array}\right].
	\end{align}
	A collision resolution algorithm is defined by matrices $X_0$, $X_1$, and $Y_n$. 
	
	We present the matrices $X_0$, $X_1$, and $Y_n$ for a typical collision resolution algorithm known as the tree splitting algorithm \cite{tree_splitting}. All reservation signals involved in the collision resolution algorithm are referred to as packets here. When a collision occurs, a new CRP begins while all packets involved in the collision are split into two subsets with equal probabilities. Any newly generated packets should wait until all packets in the current CRP have been successfully transmitted. In each timeslot, a subset of packets are transmitted. If no collision occurs, then the next subset of packets are transmitted in the next timeslot. If a collision occurs, then these involved packets are split again into two subsets to be transmitted in the next timeslots.
	
	\begin{figure}
		\centering
		\includegraphics[width=1\linewidth]{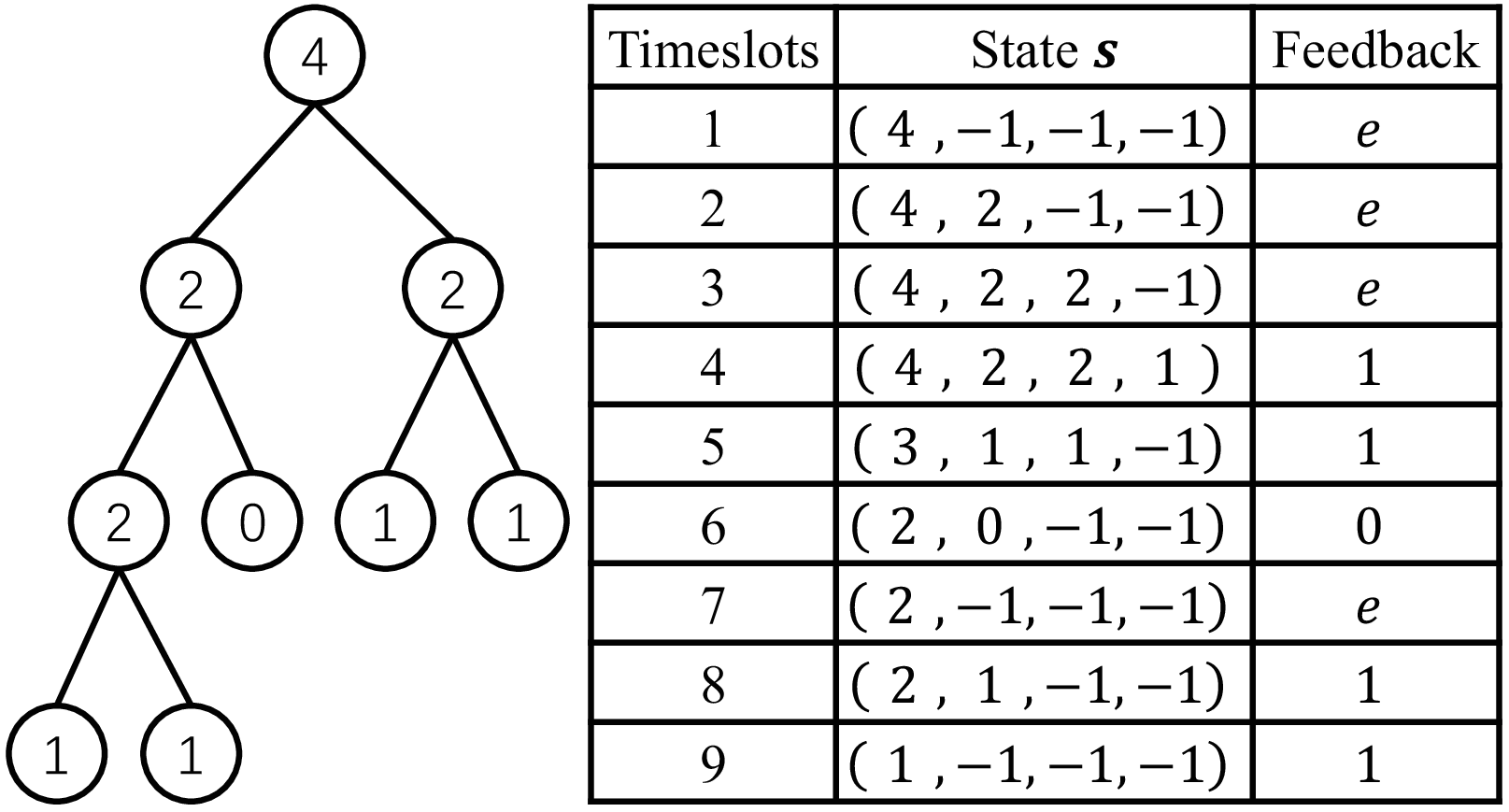}\\
		\caption{An illustration of the tree splitting algorithm. There are 4 packets in the CRP, which takes 9 timeslots to address the collision. In each timeslot, the feedback $e$, $1$, and $0$ represent a collision, a successful transmission, and an idle timeslots, respectively.}
		\label{fig_tree_splitting}
	\end{figure}
	
	The CRP is characterized by a binary tree. After the $i$th split to the $i$th layer, the number of packets in the split subset is $m_i$, for $i = 0, 1, \ldots, R$. We assume that the tree splitting algorithm is split for at most $R$ layers. Suppose the tree splitting algorithm has been split to the $x$th layer. The $m_x$ packets in the $x$ layer is sent in the current timeslot. We set $m_i = -1$ for $i > x$. Let $s = (m_0, \ldots, m_R)$ denote the state of the tree splitting. The state transition is illustrated in Fig. \ref{fig_tree_splitting}. In addition, the reservation state of each node $e_i$ is represented by the layer of the node, i.e., $e_i = x$ when the node is at the $x$th layer of the binary tree. For an idle timeslot with $m_x = 0$, the splitting layer decreases by one. The reservation state of all nodes remains unchanged. When a collision occurs with $m_x \geq 2$, then the $m_x$ packets are split into two subsets containing $m_{x+1}'$ and $m_x - m_{x+1}'$ packets with probability given by
	\begin{align}
		\frac{\binom{m_x}{m_{x+1}'}}{2^{m_x}}.
	\end{align}
	Each node at the $x$th layer may be split to the next layer $x+1$ with probability $\frac{1}{2}$. Particularly if a collision occurs when $x=R$, then the packets are moved out of the current CRP without split. These packets may join the next CRP and the reservation state of them becomes zero. If $m_x = 1$, then a packet is successfully transmitted in the current timeslot. The packet at the $x$th layer can begin its transmission. Thus, the reservation state of the node becomes zero. Therefore, for the reservation state of each node, only those at the $x$th layer of the binary tree may change the state. For a node at the $x$th layer, the transition probability of its reservation state from $e_i = x$ to $e_i'$ is given by
	\begin{align}
		p_{i,e_i,e_i'} = \begin{cases}
			1, & \mathrm{if} x = R, m_x \geq 2, e_i' = 0, \\
			1, & \mathrm{if} m_x = 1, e_i' = e_i \mathbbm{1}_{\{d > 1\}}, \\
			\frac{1}{2}, & \mathrm{if} m_x \geq 2, e_i' \in \{e_i, e_i + 1\}, \\
			0, & \mathrm{otherwise}.
		\end{cases}
	\end{align}
	Hence, in cases where no packet is successfully transmitted, we can depict the transition matrix $X_0$ governing the states of the reservation procedure. The matrix's elements are given by
	\begin{align}
		\label{eq_tree_X0}
		&X_{0,ij} = \nonumber\\
		&\begin{cases}
			1, & \mathrm{if}~ m_x = 0, m_x' = -1, m_k' = m_k ~\mathrm{for}~ k \neq x, \\
			1, & \mathrm{if}~ x = R, m_x \geq 2, m_k' = m_k - m_x, m_x' = -1, \\
			\frac{\binom{m_x}{m_{x+1}'}}{2^{m_x}}, & \mathrm{if}~ x < R, m_x \geq 2, m_k' = m_k ~\mathrm{for}~ k \neq x+1, \\
			0, & \mathrm{otherwise},
		\end{cases}
	\end{align}
	which represents the transition probability of reservation procedure state from $s_i = (m_0, \ldots, m_R)$ to state $s_j = (m_0', \ldots, m_R')$. 
	
	If $m_x = 1$, then a packet is successfully transmitted in the current timeslot and the splitting layer decreases by one. The number of packets $m_k$ for layers $k < x$ is reduced by one. Thus, the elements of the transition matrix $X_1$ is given by
	\begin{align}
		X_{1, ij} = \begin{cases}
			1, & \mathrm{if}~ m_x = 1, m_k' = -1 ~\mathrm{for}~ k \geq x, \\
			& \quad~ m_{\ell}' = m_{\ell} - 1 ~\mathrm{for}~ \ell < x, \\
			0, & \mathrm{otherwise}.
		\end{cases}
	\end{align}
	Note that the state of the reservation procedure characterizes the whole system while the reservation state of each node characterizes individual node information. Thus, in this case the total amount of the layer number for all nodes are obtained by $\bm{e}$ or $s$. The states satisfy that 
	\begin{align}
		\sum_{i=1}^{x} m_i = \sum_{i=1}^{M} e_i, 
	\end{align}
	otherwise the state is not a feasible state in the state space. 
	
	At the beginning of a timeslot, a new CRP begins if all packets in the current CRP has been successfully transmitted with $m_0 = -1$, otherwise the state $s$ remains unchanged. Therefore, when the number of packets in the reservation queue is $n$, the transition matrix $Y_n$  is given by
	\begin{align}
		Y_{n,ij} = \begin{cases}
			1, & \mathrm{if}~ m_0 = -1, m_0' = n, m_k' = -1 ~\mathrm{for}~k>0, \\
			1, & \mathrm{if}~ m_0 \geq 0, m_k' = m_k ~\mathrm{for}~ k\geq0, \\
			0, & \mathrm{otherwise}.
		\end{cases}
	\end{align}

	According to the definition of the state for the tree splitting algorithm, it satisfies that $m_i \leq m_j$ for $i > j$, and $-1 \leq m_i \leq N$. The size of the state space of $s$ is given by $r_0 = \binom{N+R+2}{R}$. For the state of the whole system, the transition probability in the second step is the same as that under the polling scheme. In the first step, the transition probability of $s$ and $e_i$ is given as above. As for state $t_d$ and $q_i$ in the first step, the next state $d'$ and $q_i'$ for $i = 1, \ldots, M$ are given by
	\begin{align}
		t_d' &= \begin{cases}
			c, & t_d = 0, m_x = 1, \\
			0, & t_d = 0, m_x \neq 1, \\
			t_d - 1, & t_d > 0.
		\end{cases} \\
		q_i' &= \begin{cases}
			0, & t_d = 1, e_i = x, \\
			q_i, & \mathrm{otherwise}.
		\end{cases}
	\end{align}
	
	Therefore, we can derive the transition probability matrix $P$. The steady-state probability is computed by Eq. (\ref{eq_steady_state_polling}). Based on that, the average AoII of all nodes is given by
	\begin{align}
		\bar{\Delta} = \sum_{q_1=0}^{N} \cdots \sum_{q_M=0}^{N} \sum_{i=1}^{M}\frac{q_i}{M}  \sum_{t_d=0}^{c} \sum_{s,e_1,\ldots,e_M} \pi_{s,e_1,\ldots,e_M,q_1,\ldots,q_M,t_d},
	\end{align}
	in which $\sum\nolimits_{s,e_1,\ldots,e_M}$ represents the summation for all state $s, e_1, \ldots, e_M$.

	\section{Reduced-Dimensional Markov Model for Average Peak AoII Minimization}
	
	\begin{figure}
		\centering
		\includegraphics[width=1\linewidth]{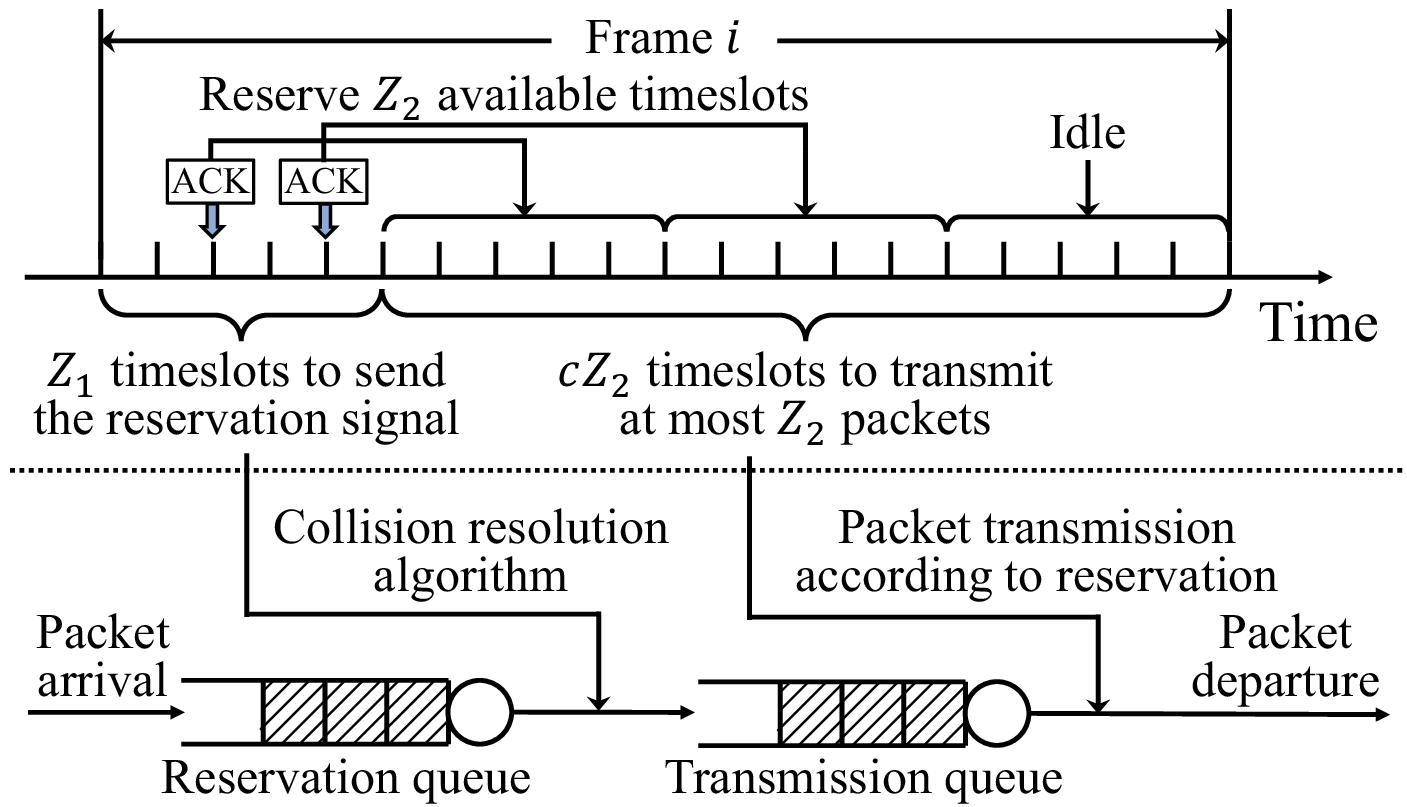}\\
		\caption{An illustration of the structure of the frame.}
		\label{system_model_TD}
	\end{figure}
	
	In this section, we formulate the Markov chain towards the peak AoII metric. All data packets are transmitted to the receiver without eviction. The peak AoII analysis for the polling scheme with various characteristics has been studied in \cite{polling_analysis_queuing}. The existing methods and results can be adopted in the presented unified framework. Thus, we focus on the random access scheme in this section. In the peak-AoII oriented scenarios, we can formulate the Markov chain based on the total number of nodes in the reservation queue and transmission queue. Thus, the dimension does not increase with the number of nodes in the system. The dimension of the Markov chain is significantly reduced. We consider three multiplexing schemes of the reservation signals and data packets under the unified framework including TD, FD, and XD schemes. 
	
	\subsection{Time-Division Multiplexing}
	

	We first consider the TD scheme, which is the basic structure for reservation-based random access in slotted systems. We consider the threshold for triggering reservation is $K=1$ for the access trigger stage in this section. The transmission constraint is also $N_{\mathrm{max}} = 1$. The time axis is split into frames, which contains a fixed number of timeslots, as shown in Fig. \ref{system_model_TD}. The reservation signal and data packets are multiplexed in timeslots of each frame. Specifically, in the first $Z_1$ timeslots of the frame, the nodes in the reservation queue send reservations signals. In the next $cZ_2$ timeslots of the frame, the nodes in the transmission queue transmit data packets. Thus, at most $Z_1$ reservation signals and $Z_2$ data packets can be transmitted in each frame. The principal formulation methods for the Markov chain of TD scheme can be found in our previous work \cite{conference_TD}. We provide the detailed justification of this part under the unified framework in Appendix \ref{appendix_TD}.

	\subsection{Frequency-Division Multiplexing}
	
	\begin{figure}[!t]
		\centering
		\includegraphics[width=1\linewidth]{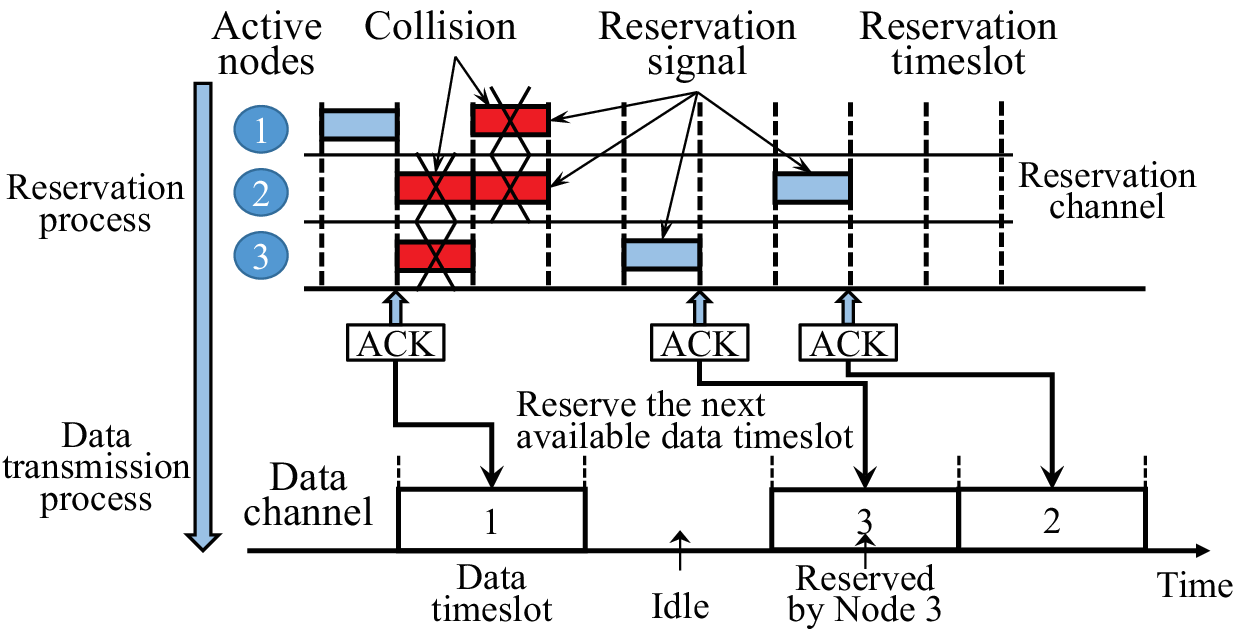}\\
		\caption{The frequency division allocation scheme.}
		\label{system_model_FD}
	\end{figure}
	
	We next consider the FD scheme that multiplexing the reservation signal and data packets in the frequency domain, as shown in Fig. \ref{system_model_FD}. Since the reservation signals and data packets are transmitted with separate spectrum, we consider a reservation channel and a data channel allocated with a fixed amount of bandwidth $w_1$ and $w_2$, respectively. Thus, the reservation signals are transmitted in the reservation channel while the data packets are transmitted in the data channel. For the access trigger stage and transmission stage, we consider $K = N_{\mathrm{max}}$. In other words, each successfully transmitted reservation signal can reserve a time interval of length $KT_2$ in the data channel for packet transmission. We denote the length of a reservation timeslot and a data timeslot by $T_1$ and $KT_2$, respectively.
	
	For the reservation stage, we consider the simple Aloha algorithm in this section. Each node sends the reservation signal in a timeslot with probability $\beta_i$ when there are $i$ reservation signals in the reservation queue. Here, we assume that all nodes know the number of reservation signals in the reservation queue. Thus, the probability that a reservation signal is successfully transmitted in a timeslot when there are $i$ reservation signals in the reservation queue is given by
	\begin{align}
		\gamma_i = i \beta_i (1 - \beta_i)^{i - 1}.
	\end{align}
	The optimal $\beta_i$ maximizing the successful reservation probability $\gamma_i$ is given by $\beta_i = \frac{1}{i}$. Moreover, we assume that the number of nodes in the reservation queue is known by the receiver and all nodes through collision level estimation. Specifically, the received energy level at the receiver increases linearly with the number of reservation signals sent by nodes. Thus, the number of nodes sending reservation signals can be estimated based on the received energy level. With the knowledge of the number of nodes in the reservation queue, we can apply more effective access control schemes for the reservation-based random access. Therefore, we consider each node sends the reservation signal with probability $\beta_i$. The probability that a reservation signal is successfully transmitted is given by
	\begin{align}
		\gamma_i = \begin{cases}
			\left(1 - \frac{1}{i}\right)^{i-1}, & i \geq 2, \\
			1, & i = 1, \\
			0, & i = 0,
		\end{cases}
	\end{align}
	where $\gamma_0 = 0$ since no reservation signal is transmitted when there is no reservation signal in the reservation queue. Particularly, for this type of Aloha, the state of collision resolution is the same as the length of the reservation queue since each node sends the reservation signal with a probability determined by the length of the reservation queue. Thus, for simplicity, the state of the tandem queue here is represented by $(q_2, q_1)$ including the length of the reservation queue and the transmission queue.
	
	We define the length of each frame as $KT_2$, which is the time for transmitting a combined packet. We focus on the state of the tandem queue at the beginning of each data timeslot and formulate a Markov chain model. The state transition depends on the packet arrival and reservation process during a data timeslot. If there are $y$ new reservation signals and $z$ successful transmission of reservation signals in a data timeslot, the state transition from state $(q_2, q_1)$ to state $(q_2', q_1')$ is represented by
	\begin{align}
		\label{eq_transition_yz}
		q_2' &= \begin{cases}
			q_2 - 1 + z, & q_2 > 0, \\
			z, & q_2 = 0,
		\end{cases} \\
		q_1' &= q_1 + y - z.
	\end{align}
	
	Next, we present the probability distribution for the number of new reservation signals $y$ and the number of successful transmission of reservation signals $z$. First consider the case when $x = \frac{KT_2}{T_1}$ is an integer. In this case, each frame contains $x$ timeslots. The number of new reservation signals in the $i$th timeslot of a frame is $y_i$, which follows Poisson distribution with mean $\frac{\lambda T_1}{K}$ since $K$ arrived data packets lead to a new reservation signal with a combined packet. Thus, the probability that there are $y_1, \ldots, y_x$ new reservation signals in each timeslots of the frame is given by
	\begin{align}
		g_{y_1, \ldots, y_x} = e^{-\frac{\lambda T_1 x}{K}} \prod_{i=1}^{x}\frac{\left(\frac{\lambda T_1}{K}\right)^{y_i}}{y_i!}.
	\end{align}
	
	Let $z_i \in \{0,1\}$ denote the number of successfully transmitted reservation signals in the $i$th timeslot of the frame. The probability of $z_i=1$ and $z_i=0$ are $\gamma_k$ or $1-\gamma_k$, respectively, when the length of the reservation queue is $k$. Given that the length of the reservation queue is $q_1$ at the beginning of the frame, the length of the reservation queue at the $i$th timeslot of the frame is given by $k_i = q_1 + \sum_{j=1}^{i-1}(y_j - z_j)$. Thus, we can derive the probability that there are $z_1, \ldots, z_x$ successfully transmitted reservation signals in each timeslot of the frame conditioned on $q_1$ and $y_1, \ldots, y_x$, given by
	\begin{align}
		f_{z_1, \ldots, z_x | q_1, y_1, \ldots, y_x} = \prod_{i=1}^{x}\left[z_i \gamma_{k_i} + (1 - z_i)(1 - \gamma_{k_i})\right].
	\end{align}
	
	Next, we consider general cases when $x = \frac{KT_2}{T_1}$ is a non-integer. A frame may contains $\left\lfloor\frac{KT_2}{T_1}\right\rfloor$ or $\left\lfloor\frac{KT_2}{T_1}\right\rfloor + 1$ timeslots. We assume that the number of timeslots within each frame is independent. Specifically, each frame contains $\left\lfloor\frac{KT_2}{T_1}\right\rfloor$ timeslots with probability
	\begin{align}
		\sigma = 1 - \frac{KT_2}{T_1} + \left\lfloor\frac{KT_2}{T_1}\right\rfloor.
	\end{align}
	Each frame contains $\left\lfloor\frac{KT_2}{T_1}\right\rfloor + 1$ timeslots with probability $1 - \sigma$. Let $\underline{x} = \left\lfloor\frac{KT_2}{T_1}\right\rfloor$ and $\overline{x} = \left\lfloor\frac{KT_2}{T_1}\right\rfloor + 1$. When the length of the reservation queue is $q_1$, we can derive the probability that the total number successfully transmitted reservation signals is $z = \sum_{i=1}^{x}z_i$ and the total number of new reservation signals is $y = \sum_{i=1}^{x}y_i$ in a frame, given by
	\begin{align}
		\label{eq_h_yz}
		&h_{y,z|q_1} \nonumber\\
		& =(1 - \sigma) \sum_{\sum_{i=1}^{\overline{x}} y_i = y} g_{y_1, \ldots, y_{\overline{x}}} \sum_{\sum_{i=1}^{\overline{x}} z_i = z} f_{z_1, \ldots, z_{\overline{x}} | q_1, y_1, \ldots, y_{\overline{x}}} \nonumber\\
		& ~~~~ + \sigma \sum_{\sum_{i=1}^{\underline{x}} y_i = y} g_{y_1, \ldots, y_{\underline{x}}} \sum_{\sum_{i=1}^{\underline{x}} z_i = z} f_{z_1, \ldots, z_{\underline{x}} | q_1, y_1, \ldots, y_{\underline{x}}}. 
	\end{align}
	In Eq. (\ref{eq_h_yz}), we use the notation $\sum\limits_{\sum_{i=1}^{\underline{x}} y_i = y}g_{y_1, \ldots, y_{\underline{x}}}$ represents the summation of $g$ for all $y_1, \ldots, y_{\underline{x}}$ in the set $\left\{y_1, \ldots, y_{\underline{x}} | \sum_{i=1}^{\underline{x}} y_i = y\right\}$, which is of space size $\binom{y + \underline{x} + 1}{y}$. The representation is similar for the four summation notations in Eq. (\ref{eq_h_yz}). For simplicity, let $(x)^{+}$ represent $\max\{x, 0\}$. When the state transits from $(q_2, q_1)$ to $(q_2', q_1')$, the values $y$ and $z$ are obtained according to Eq. (\ref{eq_transition_yz}), given by $y = q_1' + q_2' - q_1 - (q_2-1)^{+}$ and $z = q_2' - (q_2-1)^{+}$. Therefore, the transition probabilities are given by
	\begin{align}
		p_{(q_2,q_1),(q_2',q_1')} = h_{q_1' + q_2' - q_1 - (q_2-1)^{+}, q_2' - (q_2-1)^{+} | q_1}.
	\end{align}
	
	Moreover, the state satisfies that $q_2' \leq N_2$ and $q_1' \leq N_1$ due to the constraint of queue length for the tandem queue. If $q_2' > N_2$ or $q_1' > N_1$, we just drop those data packets exceeding $N$ in the reservation queue and the transmission queue to obtain a finite state space, then the state transits to $q_2' = N_2$ or $q_1' = N_1$ at the beginning of the next timeslot. The transition probabilities under the constraint of queue length $N$ are given by
	\begin{align}
		\label{eq_p_prime}
		p_{(q_2, q_1), (q_2', q_1')}' \!=\! \begin{cases}
			p_{(q_2, q_1), (q_2', q_1')}, &q_1' < N_1, q_2' < N_2 \\
			\sum\limits_{j = N_1}^{\infty}p_{(q_2, q_1), (q_2', j)}, & q_1' = N_1, q_2' < N_2 \\
			\sum\limits_{i = N_2}^{\infty}p_{(q_2, q_1), (i, q_1')}, & q_1' < N_1, q_2' = N_2 \\
			\sum\limits_{i = N_2}^{\infty}\sum\limits_{j = N_1}^{\infty}p_{(q_2, q_1), (i, j)}, & q_1' = N_1,  q_2' = N_2
		\end{cases}
	\end{align}
	With elements given by Eq. (\ref{eq_p_prime}), we can obtain the transition matrix $P$ of order $r_3 = (N_1+1)(N_2+1)$. Similar to the TD scheme, we can obtain the steady-state probabilities $\bm{\pi}$ of all states, which satisfy
	\begin{align}
		\begin{cases}
			\bm{\pi} P = \bm{\pi}, \\
			\bm{\pi} \bm{1}_{r_3} = 1.
		\end{cases}
	\end{align}
	Then we can obtain the steady-state probabilities $\bm{\pi}$ by
	\begin{align}
		\label{eq_pi_fd}
		\bm{\pi} = \bm{1}_{r_3}^{\mathrm{T}} \left(P - I_{r_3} + \bm{1}_{r_3\times r_3}\right)^{-1}. 
	\end{align}
	The steady-state probability for state $(q_2, q_1)$ is denoted by $\pi_{q_2,q_1}$. We can obtain the average length of the tandem queue given by
	\begin{align}
		\bar{L} = \sum_{q_1=0}^{N_1} \sum_{q_2=0}^{N_2} (q_1 + q_2)\pi_{q_2,q_1}.
	\end{align}
	Moreover, in addition to the peak AoII caused in the transmission queue, there exists an extra time elapse from the successful transmission of a reservation signal to the beginning of the next frame since we focus on the state at the beginning of each frame. This time elapse ranges from zero to $KT_2$. Thus, this time elapse is approximated by $\frac{KT_2}{2}$, which increases the peak AoII. The arrival rate of the combined packet is $\frac{\lambda}{K}$. Thus, the average peak AoII for the tandem queue is given by
	\begin{align}
		\label{eq_peak_aoii_FD}
		\ell = \frac{KT_2}{2} + \frac{K}{\lambda}\sum_{q_1=0}^{N_1} \sum_{q_2=0}^{N_2} (q_1 + q_2)\pi_{q_2,q_1}.
	\end{align}
	
	Intuitively, if more bandwidth are allocated for the reservation channel or the data channel, then the average peak AoII for the reservation queue or transmission queue is reduced. However, the total bandwidth for the two channels is constrained. To minimize the average peak AoII for the tandem queue, we optimize the bandwidth $w_1$ and $w_2$ for the reservation channel and data channel. Without loss of generality, we consider the total bandwidth constraint is normalized as $w_1 + w_2 = 1$.
	
	In order to assure the finite average peak AoII, we should stabilize the tandem queue. To achieve this goal, the transmission rate of the reservation channel and the data channel should be greater than the packet arrival rate. In the following, we provide the condition for finite peak AoII. First, for the reservation channel, the maximal transmission rate is given by
	\begin{align}
		\lim_{i\rightarrow\infty}q_i = e^{-1},
	\end{align}
	which means that at most $e^{-1}$ reservation signals on average can be successfully transmitted in the reservation channel. The transmission rate of the reservation channel is given by $e^{-1}w_1$. The arrival rate of the reservation signals to the reservation queue is given by $\frac{\lambda}{K}$. Thus, we have $e^{-1}w_1 > \frac{\lambda}{K}$. Second, for the data channel, the maximal transmission rate is given by $w_2 = 1 - w_1$. The packet arrival rate is $\lambda c$ since the size of each packet is $c$. It requires that $\lambda c < 1 - w_1$. Therefore, the bandwidth $w_1$ is constrained by
	\begin{align}
		\label{eq_range_w1}
		\frac{\lambda e}{K} < w_1 < 1 - \lambda c.
	\end{align}
	Furthermore, we can derive the feasible range of arrival rate $\lambda$ to attain finite peak AoII with a feasible $w_1$. Thus, the inequality $\frac{\lambda e}{K} < 1 - \lambda c$ should be satisfied. In other words, we have the following upper bound for $\lambda$ given by
	\begin{align}
		\label{eq_range_lambda_fd}
		\lambda < \frac{K}{e + Kc}.
	\end{align}
	
	When Eq. (\ref{eq_range_lambda_fd}) holds, the tandem queue is stable. In this case, we can obtain the average peak AoII from Eq. (\ref{eq_peak_aoii_FD}). Based on Eq. (\ref{eq_peak_aoii_FD}), we can optimize the bandwidth $w_1$ and $w_2$ that minimizes the average peak AoII under the constraint $w_1 + w_2 = 1$. To find the optimal bandwidth $w_1$, we represent the average peak AoII as a function of $w_1$, denoted by $\ell(w_1)$. Specifically, we can obtain the first derivative of $\ell(w_1)$ with respect to $w_1$ by
	\begin{align}
		\label{eq_derivative_ell}
		\frac{\mathrm{d}\ell(w_1)}{\mathrm{d} w_1} = \frac{Kc}{2(1-w_1)^2} + \frac{K}{\lambda} \sum_{q_1=0}^{N_1} \sum_{q_2=0}^{N_2}(q_1+q_2)\frac{\mathrm{d}\pi_{q_2,q_1}(w_1)}{\mathrm{d} w_1}.
	\end{align}
	Let $Q(w_1) = P-I_(r_3)+\bm{1}_{r_3\times r_3}$. Based on Eq. (\ref{eq_pi_fd}), the first derivative of $\bm{\pi}$ with respect to $w_1$ in Eq. (\ref{eq_derivative_ell}) is given by
	\begin{align}
		\frac{\mathrm{d}\bm{\pi}(w_1)}{\mathrm{d} w_1} &= \bm{1}_{r_3}^{\mathrm{T}} \frac{\mathrm{d}Q(w_1)^{-1}}{\mathrm{d} w_1}  \nonumber \\
		&= -\bm{1}_{r_3}^{\mathrm{T}} Q(w_1)^{-1} \frac{\mathrm{d}Q(w_1)}{\mathrm{d} w_1} Q(w_1)^{-1},
	\end{align}
	in which $\frac{\mathrm{d}Q(w_1)}{\mathrm{d} w_1}$ is the first derivative of $Q(w_1)$ with respect to $w_1$. Each element of $\frac{\mathrm{d}Q(w_1)}{\mathrm{d} w_1}$ is determined by $\frac{\mathrm{d}h_{y,z|q_1}(w_1)}{\mathrm{d} w_1}$. First, we derive the first derivative of $g_{y_1,\ldots,y_x}(w_1)$ with respect to $w_1$. Let $\bar{\lambda} = \frac{\lambda}{Kw_1}$ denote the average number of new reservation signals arriving at the reservation queue per timeslot. We can obtain that
	\begin{align}
		\frac{\mathrm{d}g_{y_1, \ldots, y_x}(w_1)}{\mathrm{d} w_1} = \frac{\bar{\lambda}^{y} e^{-\bar{\lambda}}}{w_1\prod_{i=1}^{x}y_i!} \left(\bar{\lambda}x - y\right),
	\end{align}
	in which $y$ is the number of new reservation signals given by $y = \sum_{i=1}^{x}y_i$. Next, we can derive the first derivative of $\sigma$. Although $\sigma$ is discontinuous with respect to $w_1$, we can observe that $h_{y,z|q_1}(w_1)$ is continuous with respect to $w_1$. Thus, we can use left derivative or right derivative for discontinuous points since the left derivative and right derivative of $\sigma$ are equal. The derivative of $\sigma$ is denoted by $-\frac{Kc}{(1-w_1)^2}$. Therefore, we can derive the first derivative of $h_{y,z|q_1}(w_1)$ with respect to $w_1$ given by Eq. (\ref{eq_h_w1}). 
	\begin{figure*}
		\begin{align}
			\label{eq_h_w1}
			&\frac{\mathrm{d}h_{y,z|q_1}(w_1)}{\mathrm{d} w_1} =\nonumber\\
			& (1 - \sigma) \sum_{\sum_{i=1}^{\overline{x}} y_i = y} \frac{\mathrm{d}g_{y_1, \ldots, y_{\overline{x}}}(w_1)}{\mathrm{d} w_1} \!\!\sum_{\sum_{i=1}^{\overline{x}} z_i = z}\!\! f_{z_1, \ldots, z_{\overline{x}} | q_1, y_1, \ldots, y_{\overline{x}}} + \sigma \sum_{\sum_{i=1}^{\underline{x}} y_i = y} \frac{\mathrm{d}g_{y_1, \ldots, y_{\underline{x}}}(w_1)}{\mathrm{d} w_1}\!\!\sum_{\sum_{i=1}^{\underline{x}} z_i = z}\!\! f_{z_1, \ldots, z_{\underline{x}} | q_1, y_1, \ldots, y_{\underline{x}}} \nonumber \\
			&- \frac{Kc}{(1-w_1)^2} \sum_{\sum_{i=1}^{\underline{x}} y_i = y} \!\!g_{y_1, \ldots, y_{\underline{x}}}(w_1) \!\!\sum_{\sum_{i=1}^{\underline{x}} z_i = z}\!\!\! f_{z_1, \ldots, z_{\underline{x}} | q_1, y_1, \ldots, y_{\underline{x}}} + \frac{Kc}{(1-w_1)^2} \sum_{\sum_{i=1}^{\overline{x}} y_i = y} \!\!g_{y_1, \ldots, y_{\overline{x}}}(w_1) \!\!\!\sum_{\sum_{i=1}^{\overline{x}} z_i = z}\!\! f_{z_1, \ldots, z_{\overline{x}} | q_1, y_1, \ldots, y_{\overline{x}}}.
		\end{align}
		\hrulefill
	\end{figure*}
	Thus, we can compute all elements of $\frac{\mathrm{d}Q(w_1)}{\mathrm{d} w_1}$ and $\frac{\mathrm{d}\ell(w_1)}{\mathrm{d} w_1}$.
	
	To find the optimal bandwidth $w_1$ that minimizes the average peak AoII $\ell$, we apply a binary searching algorithm in the interval $\left(\frac{\lambda e}{K}, 1 - \lambda c\right)$. Specifically, the binary searching algorithm is described as follows. Step 0, let $w_{1,1} = \frac{\lambda e}{K}$ and $w_{1,2} = 1 - \lambda c$. Step 1, Then find the point $w_{1,0} = \frac{1}{2}(w_{1,1} + w_{1,2})$ and compute $\left.\frac{\mathrm{d}\ell(w_1)}{\mathrm{d} w_1} \right| _{w_1 = w_{1,0}}$ at point $w_{1,0}$. Step 2, If $\left.\frac{\mathrm{d}\ell(w_1)}{\mathrm{d} w_1} \right| _{w_1 = w_{1,0}} < 0$, then let $w_{1,1} = w_{1,0}$, otherwise let $w_{1,2} = w{1,0}$. Step 3, If $w_{1,2} - w_{1,1} < \epsilon$ for a given $\epsilon$, then we can obtain the optimal bandwidth $w_1^* = \frac{1}{2}(w_{1,1} + w_{1,2})$, otherwise go back to step 1.

	\subsection{Dynamic Bandwidth Allocation}
	
	Both the TD and FD schemes presented in Sections III and IV are static multiplexing schemes, as they maintain a fixed structure for multiplexing reservation signals and data packets, regardless of the tandem queue's current state. To further improve the bandwidth efficiency, we present the XD scheme in this section. Specifically, the bandwidth $w_1$ for the reservation signals and $w_2$ for the data packets are modified dynamically based on the length of the reservation queue and the transmission queue. For the access trigger stage, we consider $K=1$, while we consider $N_{\mathrm{max}} = \infty$ for the transmission stage. Thus, after a reservation signal is successfully transmitted by a node, the node can transmit all data packets in the local buffer when the node begins to transmit. The reservation stage applies the Aloha algorithm. We consider $N_1 = M$ and $0 \leq N_2 \leq M$ so that each node can trigger a reservation with one data packet in the local buffer. Further reservation signals halt when $N_2$ nodes are in the transmission queue.

	The length of each timeslot is set to be 1. The bandwidth can be $w_1 > 0$ or $w_1 = 0$. If $w_1 > 0$, each reservation signal is sent within $\frac{1}{w_1}$ timeslots, which is referred to as the reservation interval. For simplicity, we can modify the bandwidth $w_1$ at the beginning of each timeslot. Particularly, when we set $w_1 = 0$, then all bandwidth are allocated for the data packets and no reservation signal is transmitted. Let $t_d$ denote the index of timeslot within a reservation interval. Since each node may have different numbers of data packets in the local buffer, we use $q_1$ and $q_2$ to denote the number of nodes in the reservation queue and in the transmission queue, respectively. We use $q_0$ to denote the total number of data packets in all nodes except for those packets in the node at the head of the transmission queue. We use $q_3$ to denote the number of data packets in the node at the head of the data queue. In addition, let $s$ denote the state of the collision resolution algorithm. Thus, the whole state of the tandem queue is denoted by $S = (t_d,q_0,q_1,q_2,q_3,s)$.
	
	Next, we present the state transition probability for a given bandwidth. The transition is considered in two steps. First,  nodes in the reservation queue and  at the head of transmission queue can transmit in the current timeslot. Second, the timeslot's end can see data packet arrivals at nodes, and transmission queue head changes due to node completion.  Specifically, in the first step, only the nodes at the head of transmission queue can transmit. When $w_1 = 0$, at most $\frac{1}{c}$ data packets can be transmitted in a timeslot. Thus, the number of data packets at the head of the transmission queue transits from $q_3$ to $q_3'$ given by
	\begin{align}
		q_3' = \left(q_3-\frac{1}{c}\right)^{+}.
	\end{align}
	The other states do not change until the end of the timeslot. When $w_1 > 0$, both the reservation signals and data packets are allocated with bandwidth for transmission. Since a reservation interval lasts for $\frac{1}{w_1}$ timeslots, A reservation interval spans $\frac{1}{w_1}$ timeslots, during which bandwidth cannot be altered. The index $t_d$ increases from one to $\frac{1}{w_1}$ to denote the progress of a reservation interval. When $t_d < \frac{1}{w_1}$, the index $t_d$ increases by one while $q_1$, $q_2$, and $s$ does not change in the middle of a reservation interval. When $t_d = \frac{1}{w_1}$, a reservation interval is completed and the index changes to $t_d' = 1$. The state of the resolution algorithm $s$ transits based on the transition probability matrix of the collision resolution algorithm $X_0$, $X_1$, and $Y_n$. Specifically, these states transit from $(q_1,q_2,s_i)$ to $(q_1,q_2,s_j)$ with probability $X_{0,ij}$, or to $(q_1-1,q_2+1,s_j)$ with probability $X_{1,ij}$. For the transmission queue, at most $\frac{w_2}{c}$ data packets can be transmitted by the node at the head of the transmission queue. Thus, the state $q_3$ transits to $q_3' = \left(q_3-\frac{w_2}{c}\right)^{+}$. Therefore, all state transition from state $S = (t_d,q_0,q_1,q_2,q_3,s)$ to state $S' = (t_d',q_0',q_1',q_2',q_3',s')$ in the first step is given by
	\begin{align}
		\label{eq_transition_p1_dynamic}
		&P_{1,S,S',w_1} \nonumber\\
		&= \begin{cases}
			1, & w_1 = 0, \\
			\sum_{k=1}^{r_0}X_{0,ik}Y_{q_1,kj}, & t_d = \frac{1}{w_1}, t_d' = 1, s = s_i, s' = s_j, \\
			\sum_{k=1}^{r_0}X_{1,ik}Y_{q_1,kj}, & t_d = \frac{1}{w_1}, t_d' = 1, q_2' = q_2 + 1, \\
			& \quad q_1' = q_1 - 1, s = s_i, s' = s_j, \\
			1, & w_1 > 0, t_d < \frac{1}{w_1}, t_d' = t_d + 1, \\
			0, & \mathrm{otherwise}.
		\end{cases}
	\end{align}
	All but the last case in Eq. (\ref{eq_transition_p1_dynamic}) satisfies the conditions that $q_3' = \left(q_3-\frac{1-w_1}{c}\right)^{+}$ and other states remain unchanged if not specified.
	
	The second step does not depend on the bandwidth. In the second step of transition, the state transits from $S' = (t_d',q_0',q_1',q_2',q_3',s')$ to state $S'' = (t_d'',q_0'', q_1'', q_2'', q_3'', s'')$. If the number of data packets in the node at the head of the transmission queue $q_3'$ is greater than zero, then the state $q_2'$ and $q_3'$ does not change. However, if $q_3' = 0$, then data packets of the next node in the transmission queue moves to the head starting transmission in the next timeslot. Since the packet arrivals are independent among all nodes, we assume that each arrived packet is regarded to be hold by all nodes with equal probability for simplicity. Thus, when there are $q_0'$ packets in total and $q_1' + q_2'$ nodes in the tandem queue, if $q_2'>0$, then the probability that the next node with $i$ data packets moves to the head of the transmission queue is given by
	\begin{align}
		u_{1,q_0',q_1',q_2',i} =& \binom{q_0' - q_1' - q_2'}{i - 1} \left(\frac{1}{q_1'+q_2'}\right)^{i-1} \nonumber\\
		&\times \left(1-\frac{1}{q_1'+q_2'}\right)^{q_0' - q_1' - q_2' - i}.
	\end{align}
	If $q_2'=0$, then there are no more packets in the transmission queue. In this case, $q_3'' = 0$. Since the next node in the transmission queue starts transmission in the next timeslot, we can subtract those packets from the total number of packets remained in the tandem queue. The state $q_0'$ decreases by $i$ and $q_2'$ decreases by one becoming $q_2'' = q_2' - 1$.
	
	Moreover, new packets may arrive at nodes. The total number of data packets at all nodes $q_0'$ increases by the number of newly arrived packets. According to the Poisson arrival, there are $y$ newly arrived packets with probability $a_{y}$. Each packet arrive at each node with equal probability. If a packet arrives at a node that does not have packets in the buffer currently, then the node enters the reservation queue to send reservation signals. The state $q_1'$ transits to $q_1' + j$ if there are $j$ new nodes entering the reservation queue. Since there are $M$ nodes in total with $q_1' + q_2'$ nodes that already have at least one packets in the buffer, the probability that there are $j$ new reservation signals is given by
	\begin{align}
		u_{2,M,q_1',q_2',y,j} = \frac{\binom{M-q_1'-q_2'+j-1}{M-q_1'-q_2'-1}\binom{q_1'+q_2'+y-j-1}{q_1'+q_2'-1}}%
		{\binom{M+y-1}{M-1}}.
	\end{align}
	In the second step, state $t_d$ and $s$ does not change. Therefore, the state transition probability from state $S' = (t_d',q_0',q_1',q_2',q_3',s')$ to state $S'' = (t_d'',q_0'', q_1'', q_2'', q_3'', s'')$ is given by Eq. (\ref{eq_transition_p2_dynamic}). 
	\begin{figure*}
		\normalsize
		\begin{align}
			\label{eq_transition_p2_dynamic}
			&P_{2,S',S''} = \begin{cases}
				a_y u_{1,q_0',q_1',q_2',i} u_{2,M,q_1',q_2',y,j}, & q_3' = 0, q_2' > 0, q_3'' = i, q_0'' = q_0' + y + i, q_1'' = q_1' + j, q_2'' = q_2' - 1, \\
				a_y u_{2,M,q_1',q_2',y,j}, & q_2' = 0, q_0'' = q_0' + y, q_1'' = q_1' + j, \\
				a_y u_{2,M,q_1',q_2',y,j}, & q_3' > 0, q_0'' = q_0' + y, q_1'' = q_1' + j, \\
				0, & \mathrm{otherwise}.
			\end{cases}
		\end{align}
		\hrulefill
	\end{figure*}
	
	The overall transition probability from state $S$ at the beginning of a timeslot to state $S''$ at the beginning of the next timeslot is obtained by multiplying the transition probability of two steps given in Eqs. (\ref{eq_transition_p1_dynamic}) and (\ref{eq_transition_p2_dynamic}). We next show that for each pair of states $S$ and $S''$, there exists at most one intermediate state $S'$. In other words, $S'$ is determined by $S$ and $S''$. Specifically, state $S'$ satisfies that $t_d' = t_d''$, $q_0' = q_0$, $q_3' = \left(q_3-\frac{1-w_1}{c}\right)^{+}$, $s' = s''$. If $q_3'' = 0$, then $q_2'=0$ since there is no nodes in the transmission queue. If $q_3'' > 0$ and $q_3' = 0$, then $q_2' = q_2'' + 1$ since the node at the head of the transmission queue finishes transmission in the current timeslot. If $q_3'' > 0$ and $q_3' = q_3''$, then $q_2' = q_2''$. For $q_1'$, it satisfies that $q_1' = q_1 + q_2 - q_2'$. Therefore, we can obtain a unique $S'$ as the intermediate state between two states $S$ and $S''$ in the consecutive timeslots. The transition probability from state $S$ to $S''$ is given by
	\begin{align}
		P_{3,S,S'',w_1} = P_{1,S,S',w_1} P_{2,S',S''}.
	\end{align}
	
	Our goal is to minimize the average peak AoII, which is equivalent to minimizing the average queue length of the tandem queue according to Eq. (\ref{eq_average_delay}). Therefore, the cost function for each timeslot is the total number of data packets in the tandem queue. Specifically, we define the cost function as $C_{S,w_1}$ given by
	\begin{align}
		C_{S,w_1} = q_0 + \left(q_3-\frac{1-w_1}{c}\right)^{+},
	\end{align}
	which represents the number of packets left in the tandem queue after the transmission in the current timeslot. The average cost per timeslot is equal to the average length $\bar{L}$ of the tandem queue.
	
	We consider that the available actions are $\frac{1}{i}$ for $i = 1, \ldots, i_{\mathrm{max}}$ to make the reservation interval contain $i$ timeslots. The action space $\mathcal{A}_{S}$ depends on the state $S$. Since we modify the bandwidth $w_1$ only at the beginning of each reservation interval, the current bandwidth is also included in the system state. Thus, the state is given by $(t_d,q_0,q_1,q_2,q_3,s,w_1)$, in which $w_1$ is only used to constrain the available action in the current timeslot when $t_d > 1$. We can formulate an infinite horizon MDP for the dynamic bandwidth scheme with the state space $\mathcal{S}$ of states $(t_d,q_0,q_1,q_2,q_3,s,w_1)$, action space $\mathcal{A}_{S}$, transition probabilities $P_{3,S,S',w_1}$, and the cost function $C_{S,w_1}$. Let $v(S)$ denote the value function of state $S$ for $S \in \mathcal{S}$. The optimal value function for the MDP satisfies the optimality equation, which is also known as Bellman Equation, given by
	\begin{align}
		\label{eq_optimality_dynamic}
		\bar{L} + v(S) = \min_{w_1 \in \mathcal{A}}\left\{C_{S,w_1} + \sum_{S' \in \mathcal{S}}P_{3,S,S',w_1}v(S')\right\}.
	\end{align}
	By solving the optimality equation (\ref{eq_optimality_dynamic}), we can obtain the optimal value function, which gives the optimal dynamic bandwidth allocation policy minimizing the average queue length.

	\begin{algorithm}[!t]
		\caption{Value Iteration Algorithm}
		Initialize $v_0\leftarrow0, \varepsilon>0, i\leftarrow 0$
		\begin{algorithmic}
			\label{algorithm_value_iteration}
			\REPEAT
			\STATE $i\leftarrow i+1$.
			\FORALL{$S\in\mathcal{S}$}
			\STATE $v_i(S) = \sum\limits_{w_1 \in \mathcal{A}_{S}}\left\{ C_{S,w_1} + \sum\limits_{S' \in \mathcal{S}} P_{3,S,S',w_1} v_{i-1}(S')\right\}$
			\ENDFOR
			\UNTIL{$\varphi(v_i-v_{i-1})<\varepsilon$}
			\FORALL{$S\in\mathcal{S}$}
			\STATE $w_1(S) = \mathop{\arg\min}\limits_{w_1 \in \mathcal{A}_{S}} \left\{C_{S,w_1} + \sum\limits_{S' \in \mathcal{S}} P_{3,S,S',w_1} v_i(S')\right\}$
			\ENDFOR
		\end{algorithmic}
	\end{algorithm}
	
	We apply the value iteration algorithm to solve the MDP \cite[Chapter 8]{MDP_book}. The optimal value function $v(S)$ is computed iteratively, as shown in Algorithm \ref{algorithm_value_iteration}.
	Through iteration computation of $v_i(S)$ for all states $S \in \mathcal{S}$, the value function converges to the optimal one. The stopping rule of the iteration is given by $\varphi(v_i - v_{i-1}) < \varepsilon$, in which $v_i$ represents all the values $v_i{S}$ for state $S \in \mathcal{S}$ and $\varphi$ is defined as
	\begin{align}
		\varphi(v_i - v_{i-1}) =& \max_{S \in \mathcal{S}}\left\{v_i(S) - v_{i-1}(S)\right\} \nonumber\\
		&- \min_{S \in \mathcal{S}}\left\{v_i(S) - v_{i-1}(S)\right\}.
	\end{align}
	Based on the optimal value function $v_i$ after the iteration converges, we can obtain an $\varepsilon$-optimal dynamic bandwidth allocation policy, given by
	\begin{align}
		\label{eq_policy_dynamic}
		w_1(S) = \mathop{\arg\min}_{w_1 \in \mathcal{A}_{S}} \left\{C_{S,w_1} + \sum_{S' \in \mathcal{S}} P_{3,S,S',w_1} v_i(S')\right\}.
	\end{align}
	This policy is $\varepsilon$-optimal because the gap of average cost between the policy and the theoretically optimal one is less than $\epsilon$. Moreover, we can obtain an approximation of the average queue length attained by the dynamic bandwidth allocation policy $w_1(S)$ given in Eq. (\ref{eq_policy_dynamic}), given by
	\begin{align}
		\bar{L}_{\varepsilon} = \frac{1}{2}\!\left[\max_{S \in \mathcal{S}}\left\{v_i(S) - v_{i-1}(S)\right\} + \min_{S \in \mathcal{S}}\left\{v_i(S) - v_{i-1}(S)\right\}\right]\!.
	\end{align}

	\section{Mean-Field Approximations for Massive Users}
	
	As the number of nodes in the system increases, the dimension of the presented framework can be prohibitively large. To reduce the computational complexity, we formulate the Markov chain model for a single node based on mean-field approximations. According to the multiple access framework, a node cannot successfully send a reservation signal if other nodes are sending reservation signals or transmitting data packets. Such impact of all nodes' reservation and transmission on a single node can be approximated by a simple statistical effect when the number of nodes $M$ is large enough. The dependence among users vanishes as $M\rightarrow\infty$. Specifically, we assume that each time the node sends a reservation signal, the reservation signal is successfully transmitted with a constant probability that is independent with the states of other nodes. A fixed-point iteration-based method is presented to obtain the steady-state probability as well as the average AoII and peak AoII. 
	
	We study the mean-field approximation for dynamic bandwidth allocation scheme in this section. The mean-field  approximations can also be applied to approximate the average AoII and peak AoII for other schemes. We focus on a single node and denote the state of the node by $(q,t_d)$ when the node has $q$ data packets in the local buffer and the current timeslot is the $t_d$th timeslot for transmission of the current data packet. When $t_d = 0$, it means that the node has not successfully transmitted the reservation signal. Consider that the constraint of the local buffer size is $N$, thus the node has at most $N$ data packets. The finite buffer size $N$ enables the formulation of Markov chain with finite state space. Note that when the constraint $N$ is large enough, the packet loss rate is approximately zero. The state space is given by
	\begin{align}
		\mathcal{S} = \{(0,0)\} \cup \{(q,t_d) | 1\leq q \leq N, 0 \leq t_d \leq c\}.
	\end{align}
	The state space size is $r_4 = N(c+1)+1$. Let $\bm{\pi}$ denote the steady-state probability of all states, while $\pi_{q,t_d}$ denote the steady-state probability of state $(q,t_d)$.
	
	\begin{figure}[!t]
		\centering
		\includegraphics[width=1\linewidth]{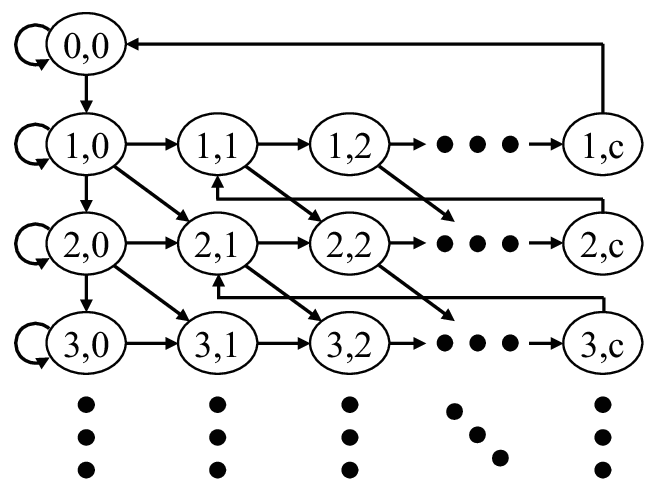}\\
		\caption{The transition of the Markov chain.}
		\label{mean_field_chain}
	\end{figure}
	
	Consider $N_2 = 0$, thus the node begins transmission of data packets right after the node successfully sends a reservation signal. Specifically, when the state of the node is $(q,0)$, then the node may successfully send a reservation signal with a probability, which we assume to be identical and independent in all timeslots. Then the node enters state $(q,1)$ and the state changes as $(q,1)$, $(q,2), \ldots, (q,c)$ in each subsequent timeslot. Then the node begins transmission of the next data packet thus the state transits from $(q,c)$ to $(q-1,1)$. The transmission process continues until the node finishes transmission of all data packets and enters state $(0,0)$. The transition of the Markov chain is shown in Fig. \ref{mean_field_chain}. In addition, new data packets may arrive at the end of each timeslot. Since the arrival rate for each node is $\bar{\lambda} = \frac{\lambda}{M}$ that approaches zero as the number of nodes $M$ increases, we assume that at most one data packet can arrive in each timeslot. Specifically, the probability that a new data packet arrive is given by $\bar{\lambda}$. 
	
	Next, we justify the probability of successfully transmitting a reservation signal in a timeslot. The node can only send reservation signals when no other node is transmitting data packets. Let $\bar{h}$ denote the expected number of data packets to be transmitted when a node successfully sends a reservation signal. Let $\gamma$ denote the throughput of the reservation signal, which represents the probability that a reservation signal of any node is successfully transmitted in a timeslot of reservation. Here we consider Aloha scheme for the reservation stage thus we have $\gamma = e^{-1}$. The number of timeslots used for reservation $k_1$ and the number of timeslots used for transmission $k_2$ satisfy that 
	\begin{align}
		k_1  \gamma  \bar{h}  c = k_2,
	\end{align}
	where $k_1 \gamma$ is the expected number of successfully transmitted reservation signals. For each successful reservation, the node transmits $\bar{h}$ data packets with total length $\bar{h} c$ in expectation, which consumes $k_2$ timeslots. Thus, the ratio of timeslots that used for reservation is given by
	\begin{align}
		\eta &= \frac{k_1}{k_1 + k_2} \nonumber\\
		&= \frac{1}{1 + \gamma \bar{h} c}.
	\end{align}
	We assume in each timeslot, the probability that there exists another node transmitting data packets is $1 - \eta$, which is independent and idential among timeslots. When the current timeslot is used for reservation with probability $\eta$, a reservation signal is successfully transmitted with probability $\gamma$. The probability that a node has data packets in the local buffer is given by $1 - \pi_{0,0}$. Thus, the expected number of nodes in the reservation stage is $M(1 - \pi_{0,0})$. Because of the symmetry among all nodes, the targeted node successfully sends a reservation signal with probability given by
	\begin{align}
		\label{eq_alpha}
		\alpha = \frac{\eta \gamma}{M (1 - \pi_{0,0})}.
	\end{align}
	Therefore, for each state $(q,0)$ where $q > 0$, the transition probability to the next state $(q',t_d')$ is given by
	\begin{align}
		\label{transition_meanfield_1}
		p_{(q,0),(q',t_d')} = \begin{cases}
			\alpha \bar{\lambda}, & q' = q + 1, t_d' = 1, \\
			\alpha (1 - \bar{\lambda}), & q' = q, t_d' = 1, \\
			(1 - \alpha) \bar{\lambda}, & q' = q + 1, t_d' = 0, \\
			(1 - \alpha) (1 - \bar{\lambda}), & q' = q, t_d' = 0. 
		\end{cases}
	\end{align}
	In other cases of $q'$ and $t_d'$, the transition probability is zero. Note that if $q' > N$, then the next state just becomes $q' = N$ and the corresponding transition probability is accumulated in the transition probability for $(N,t_d')$. For state $(0,0)$, it transits to state $(1,0)$ with probability $\bar{\lambda}$ while state $(0,0)$ with probability $1 - \bar{\lambda}$. Thus, we have 
	\begin{align}
		\label{transition_meanfield_2}
		p_{(0,0), (q',0)} = \begin{cases}
			\bar{\lambda}, & q' = 1, \\
			1 - \bar{\lambda}, & q' = 0. 
		\end{cases}
	\end{align}
	Moreover, when $t_d > 0$, the node continues transmits data packets in the current timeslot, thus the transition probability to the next state $(q', t_d')$ is given by
	\begin{align}
		\label{transition_meanfield_3}
		p_{(q,t_d),(q',t_d')} = \begin{cases}
			\bar{\lambda}, & t_d < c, q' = q + 1, t_d' = t_d + 1, \\
			\bar{\lambda}, & t_d = c, q' = q, t_d' = 1, \\
			1 - \bar{\lambda}, & t_d < c, q' = q, t_d' = t_d + 1, \\
			1 - \bar{\lambda}, & t_d = c, q' = q - 1, t_d' = 1.
		\end{cases}
	\end{align}
	All transition probabilities defined by Eqs. (\ref{transition_meanfield_1})-(\ref{transition_meanfield_3}) constitute the transition matrix $P$ for a single node with mean-field approximation. Next, we present a fixed-point iteration-based algorithm to find the steady-state probabilities $\bm{\pi}$. The transition matrix $P$ is affected by the value of $\bar{h}$, which is unknown yet. However, $\bar{h}$ is determined by the steady-state probabilities $\bm{\pi}$. The reservation action of the node is independent of the number of data packets in the local buffer. Thus, we consider that the distribution of the number of data packets when the node successfully sends a reservation signal is determined by the distribution $\bm{\pi}$ conditioned on that the number of data packets $q>0$. Specifically, $\bar{h}$ is given by
	\begin{align}
		\label{eq_bar_h}
		\bar{h} = \frac{1}{1-\pi_{0,0}}\sum_{q=1}^{N}q\sum_{t_d=0}^{c}\pi_{q,t_d}.
	\end{align}
	Therefore, we can compute the steady-state probabilities $\bm{\pi}$, the expected number of data packets transmitted in a successful access $\bar{h}$, and the successful reservation probability $\alpha$ through an iterative algorithm. First, we randomly initialize the value of $\bar{h}$ and $\bm{\pi}$. Compute $\alpha$ according to Eq. (\ref{eq_alpha}). Then in each iteration, we can derive the transition matrix of the Markov chain $P$, based on which we can compute the steady-state probabilities by
	\begin{align}
		\bm{\pi} = \bm{1}_{r_4}^{\mathrm{T}} \left(P - I_{r_4} + \bm{1}_{r_4\times r_4}\right)^{-1}. 
	\end{align}
	Also we can find new values of $\bar{h}$ and $\alpha$ according to Eqs. (\ref{eq_bar_h}) and (\ref{eq_alpha}). After a number of iterations, the value of $\bm{\pi}$ converges. Thus, we can obtain the steady-state probabilities with the mean-field approximation. The average peak AoII is given by
	\begin{align}
		\ell = \frac{1}{\bar{\lambda}}\sum_{q=1}^{N}q\sum_{t_d=0}^{c}\pi_{q,t_d}.
	\end{align}
	As for the AoII-oriented scenario, when the node only keeps the newest data packet, we formulate a Markov chain for a single node. Let $s$ denote the state of the node. When $s = 0$, it represents that the node has no data packets in the local buffer. When the node has data packet in the local buffer, the state $s>0$ is defined by the age of the data packet. When the node successfully transmits a reservation signal, then the node transmits the data packet in the next $c$ timeslots. Thus, the state becomes $-c$ after the node successfully sends a reservation signal. Then the state increases by one in each timeslot to represent the transmission procedure of the data packet. Thus, the state space is given by
	\begin{align}
		\mathcal{S} = \{-c, \ldots, 0, 1, \ldots\},
	\end{align}
	Let $\pi_s$ denote the steady-state probability of state $s \in \mathcal{S}$. 
	
	Next, we present the transition probability of the Markov chain. When the node is in state $s=0$, then the state transits to $s=1$ when a new data packet is generated with probability $\bar{\lambda}$. When $s>0$, then the state increases by one in each timeslot until a reservation signal is successfully transmitted with probability $\alpha$ while the state transits to $s = -c$. The probability $\alpha$ is given by
	\begin{align}
		\alpha = \frac{\eta \gamma}{M (1 - \sum_{s=-c}^{0}\pi_s)}.
	\end{align}
	The expected number of data packets transmitted is $\bar{h} = 1$ in this AoII-oriented scenario. When $s < 0$, then the node is transmitting the data packet hence the state increases by one in each timeslot until the state becomes $s=0$. Therefore, the transition probability from state $s$ to state $s'$ is given by
	\begin{align}
		p_{s,s'} = \begin{cases}
			\bar{\lambda}, & s = 0, s' = 1, \\
			1 - \bar{\lambda}, & s = 0, s' = 0, \\
			\alpha, & s > 0, s' = -c, \\
			1 - \alpha, & s > 0, s' = s + 1, \\
			1, & s < 0, s' = s + 1,\\
			0, & \mathrm{otherwise}.
		\end{cases}
	\end{align}
	
	Based on the transition probability, we can obtain the transition matrix $P$ for the Markov chain. We can obtain the steady-state probabilities $\bm{\pi}$ in closed form. Through each access procedure, the state of the device changes from $s = 1,\ldots$ to $s = -c, \ldots, 0$. We can find that each time the state becomes $s=1$, then the state must go through $-c, \ldots, -1$ eventually. Thus the steady-state probabilities are equal for state $s\in \{-c, \ldots, -1, 1\}$. In addition, for state $s > 1$, the state in the previous timeslot must be $s-1$. According to the transition probability, for $s>1$, we have 
	\begin{align}
		\pi_{s} &= \pi_{s-1} (1 - \alpha)\nonumber\\
		&= \pi_1 (1 - \alpha)^{s - 1}.
	\end{align}
	For $s=1$, the state in the previous timeslot must be $0$. Thus, we have 
	\begin{align}
		\pi_0 = \frac{1}{\bar{\lambda}} \pi_1.
	\end{align}
	Thus, we can compute the steady-state probabilities $\pi_1$ as follows
	\begin{align}
		\label{eq_mean_field_pi1}
		\pi_1 &= \frac{1}{\frac{1}{\alpha} + c + \frac{1}{\bar{\lambda}}}.
	\end{align}
	Substitute Eq. (\ref{eq_alpha}) into Eq. (\ref{eq_mean_field_pi1}), we can obtain that the steady-state probability $\pi_1$ is a solution to the following quadratic equation.
	\begin{align}
		-M\left(\frac{1}{\bar{\lambda}} + c\right) \pi_1^2 + \left(M + \eta \gamma c + \frac{\eta \gamma}{\bar{\lambda}}\right) \pi_1 - \eta \gamma = 0.
	\end{align}
	Thus, we can obtain the closed-form steady-state probabilities. Let $a = -M\left(\frac{1}{\bar{\lambda}} + c\right)$, $b = \left(M + \eta \gamma c + \frac{\eta \gamma}{\bar{\lambda}}\right)$, $d = \eta \gamma$. The steady-state probability $\pi_1$ is given by
	\begin{align}
		\pi_1 = \frac{-b + \sqrt{b^2 - 4ad}}{2a}.
	\end{align}
	
	When the state $s < 0$, the age of the data packet is $x + s + c + 1$, where $x$ is the age when the node successfully sends the reservation signal. The expectation of $x$ is given by
	\begin{align}
		\bar{x} = \frac{\sum_{s=1}^{\infty}s\pi_{s}}{\sum_{s=1}^{\infty}\pi_s}.
	\end{align}
	After the value of $\bm{\pi}$ converges, the average AoII is given by
	\begin{align}
		\bar{\Delta} &= \sum_{s=1}^{\infty}s\pi_{s} + \sum_{s=-c}^{-1}(\bar{x} + s + c + 1)\pi_{s} \nonumber\\
		&= \left(\frac{1}{\alpha^2} + \frac{1}{\alpha} + \frac{c(c+1)}{2}\right) \pi_1.
	\end{align}

	\section{Simulation Results}
	
	In this section, we evaluate the various schemes for the multiple access network based on the unified framework. For the AoII oriented scenario, the polling scheme and Aloha scheme are adopted for the reservation stage. For the average peak-AoII oriented scenario, Aloha is adopted for the reservation stage in FD and XD schemes. The tree splitting algorithm is adopted for the TD scheme. The maximum layers of split in the tree splitting algorithm is $R = 3$. The size of each data packet is $c = 3$. The access trigger stage is set as $K=1$ for the TD and XD schemes. The transmission constraint is set as $N_{\mathrm{max}} = 1$, $N_{\mathrm{max}} = K$, and $N_{\mathrm{max}} = \infty$ for the TD, FD, and XD schemes, respectively.
	
	\begin{figure}
		\centering
		\includegraphics[width=1\linewidth]{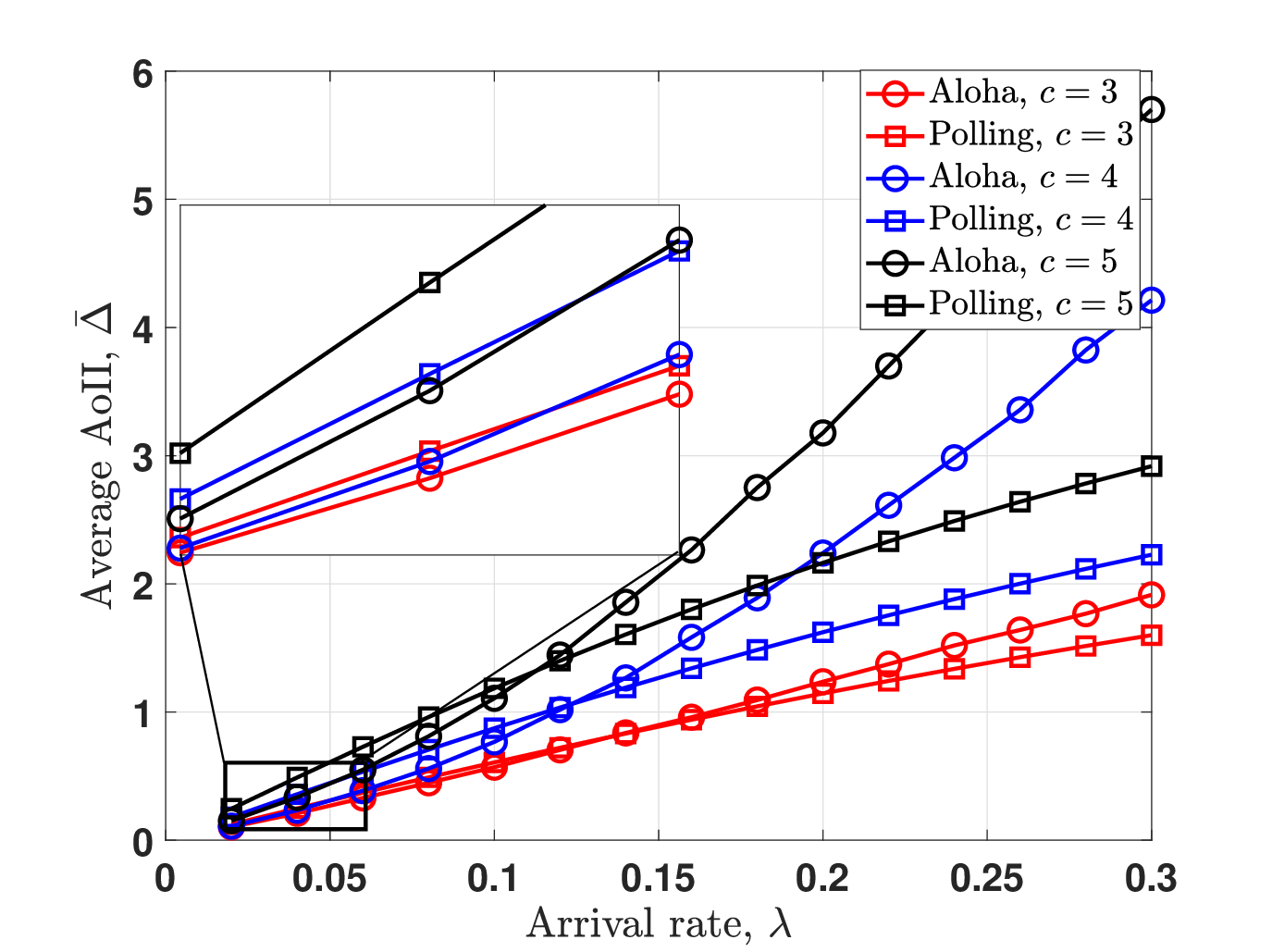}\\
		\caption{The average AoII $\bar{\Delta}$ versus the arrival rate $\lambda$ under the polling scheme and the Aloha scheme.}
		\label{fig_aoii_vs_arrival_rate}
	\end{figure}
	
	For the AoII oriented scenario, the average AoII $\bar{\Delta}$ versus the arrival rate $\lambda$ under both the polling scheme and the Aloha scheme are shown in Fig. \ref{fig_aoii_vs_arrival_rate}.We consider the number of transmitter nodes is $M=2$ and the constraint of AoII is $N = 10$. The size of each data packet is $c = 3$, $c = 4$, and $c = 5$. When the arrival rate is relatively low, the random access scheme achieves a lower average AoII. However, as the arrival rate increases, the collision among nodes under the random access scheme causes increasingly significant waste of resources, hence reducing the average AoII significantly. When the arrival rate is relatively high, the polling scheme outperforms the random access scheme under the heavy traffic. 
	
	\begin{figure}
		\centering
		\includegraphics[width=1\linewidth]{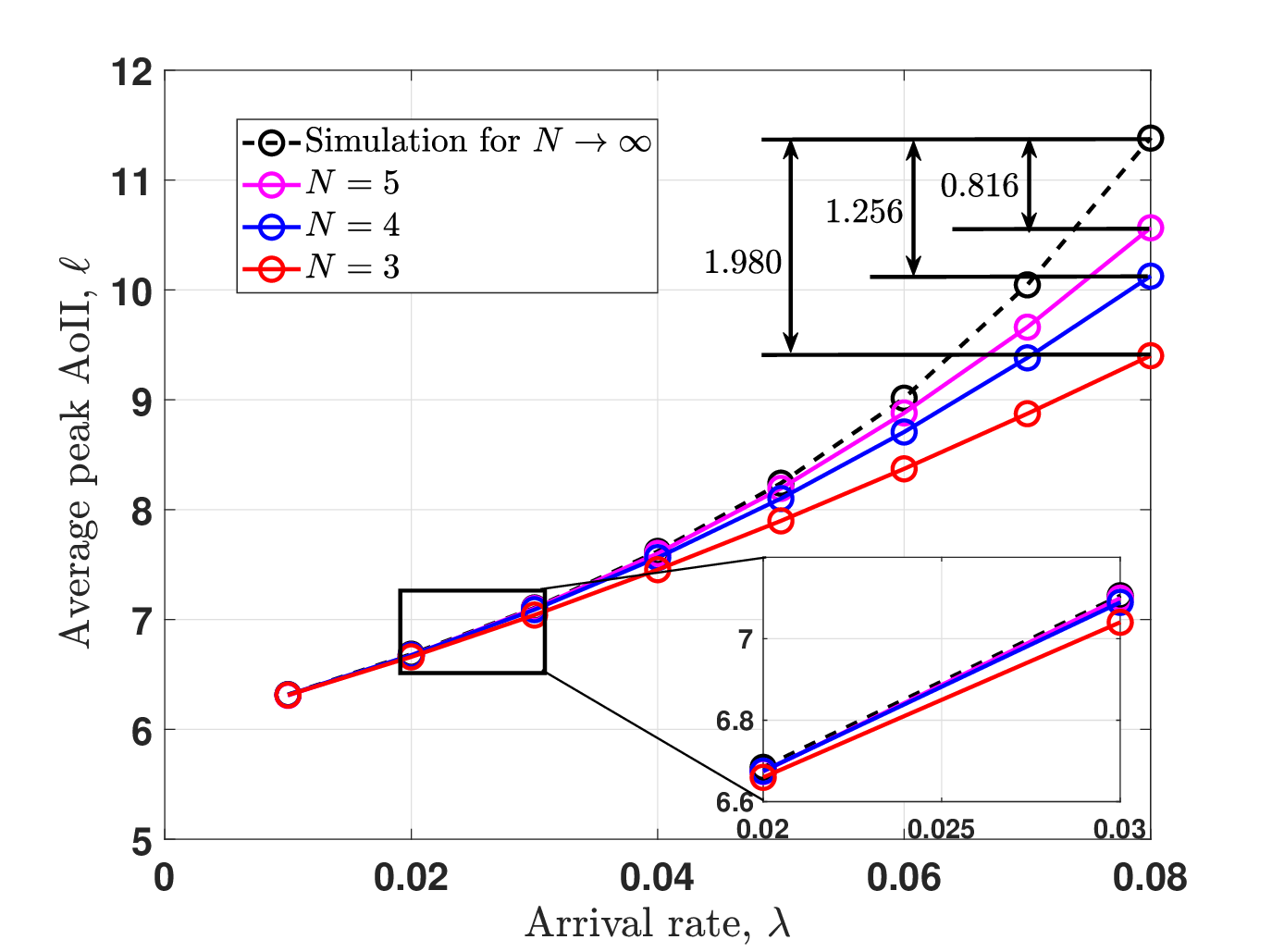}\\
		\caption{The average peak AoII $\ell$ versus the arrival rate $\lambda$ under TD scheme.}
		\label{fig_TD_delay_vs_arrival_rate}
	\end{figure}
	
	Next, for the average peak-AoII oriented scenario, we evaluate the reservation based random access scheme. For the TD scheme with the tree splitting algorithm, the average peak AoII $\ell$ versus the arrival rate $\lambda$ is shown in Fig. \ref{fig_TD_delay_vs_arrival_rate}. The structure of the frame is $Z_1 = 3$ and $Z_2 = 1$. We consider infinite-node model for this scheme since each packet requires an independent reservation. The constraint of the queue length is given by $N=3$, $N=4$, $N = 5$, and $N \rightarrow \infty$. The presented tandem queue model is applied for finite queue length while the simulation result is provided with infinite queue length constraint. The packet loss rate approaches zero when the constraint of queue length $N$ is large enough. When the arrival rate $\lambda$ is small ($\lambda \leq 0.02$), all these cases have the same average peak AoII performance and zero packet loss rate. As the arrival rate $\lambda$ increases, the gap of average peak AoII for different constraints of queue length increases.
	
	\begin{figure}
		\centering
		\includegraphics[width=1\linewidth]{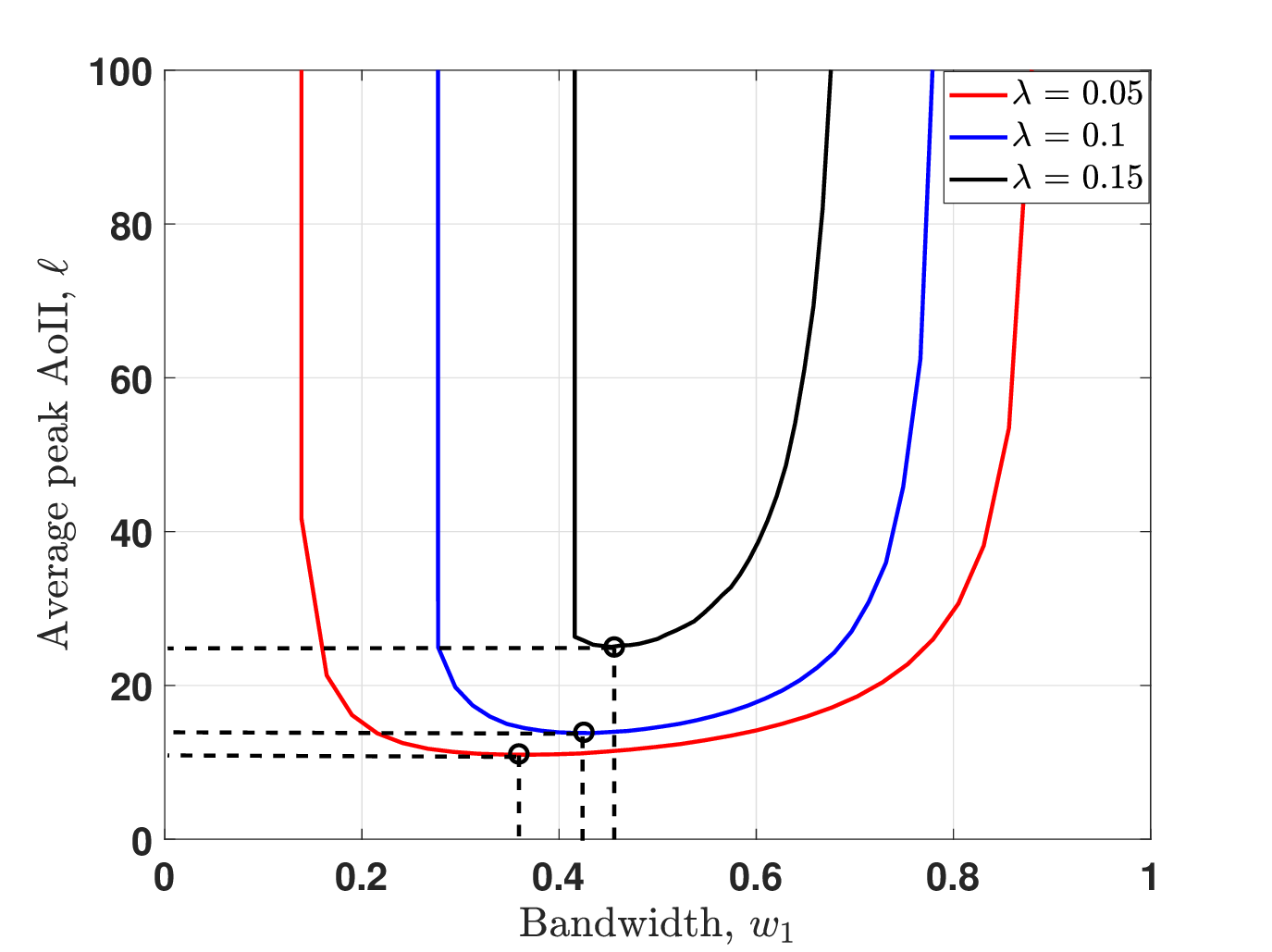}\\
		\caption{The average peak AoII $\ell$ versus the bandwidth $w_1$ under FD scheme.}
		\label{fig_FD_delay_vs_w1}
	\end{figure}
	
	For the FD scheme, we show the average peak AoII $\ell$ versus the bandwidth $w_1$ in Fig. \ref{fig_FD_delay_vs_w1}. The threshold for triggering reservation is $K = 1$. The arrival rate is given by $\lambda = 0.05$, $\lambda = 0.1$, and $\lambda = 0.15$. Since all data packets are independent with each other in reservation and transmission, the number of nodes does not affect the average in the reservation queue and transmission queue. The number of nodes only affects the peak AoII $\ell_0$ in the access trigger stage, which is zero when $K = 1$. Note that the average peak AoII becomes extremely high when the bandwidth for the reservation channel is too small or too large. Thus, we can find the optimal bandwidth allocation scheme that minimizes the average peak AoII. The optimal bandwidth $w_1$ is different with different arrival rate.
	
	Furthermore, we show the tradeoff between the average peak AoII and the supported arrival rate for the FD scheme in Fig. \ref{fig_FD_delay_vs_arrival_rate}. The number of nodes is given by $M = 5$. When $K > 1$, the average peak AoII first decreases then increases with the arrival rate since the time for waiting packet combining can be high with a low arrival rate. If the access trigger stage requires each node to wait more data packets before entering the reservation queue, then the maximum arrival rate that is supported by this scheme is improved since less reservation signals are transmitted. However, the average peak AoII can increase rapidly with a larger $K$ at a large arrival rate. It provides much insight that we can transmit different number of data packets for different scenarios in an adaptive way, which indicates the presented XD scheme.
	
	\begin{figure}
		\centering
		\includegraphics[width=1\linewidth]{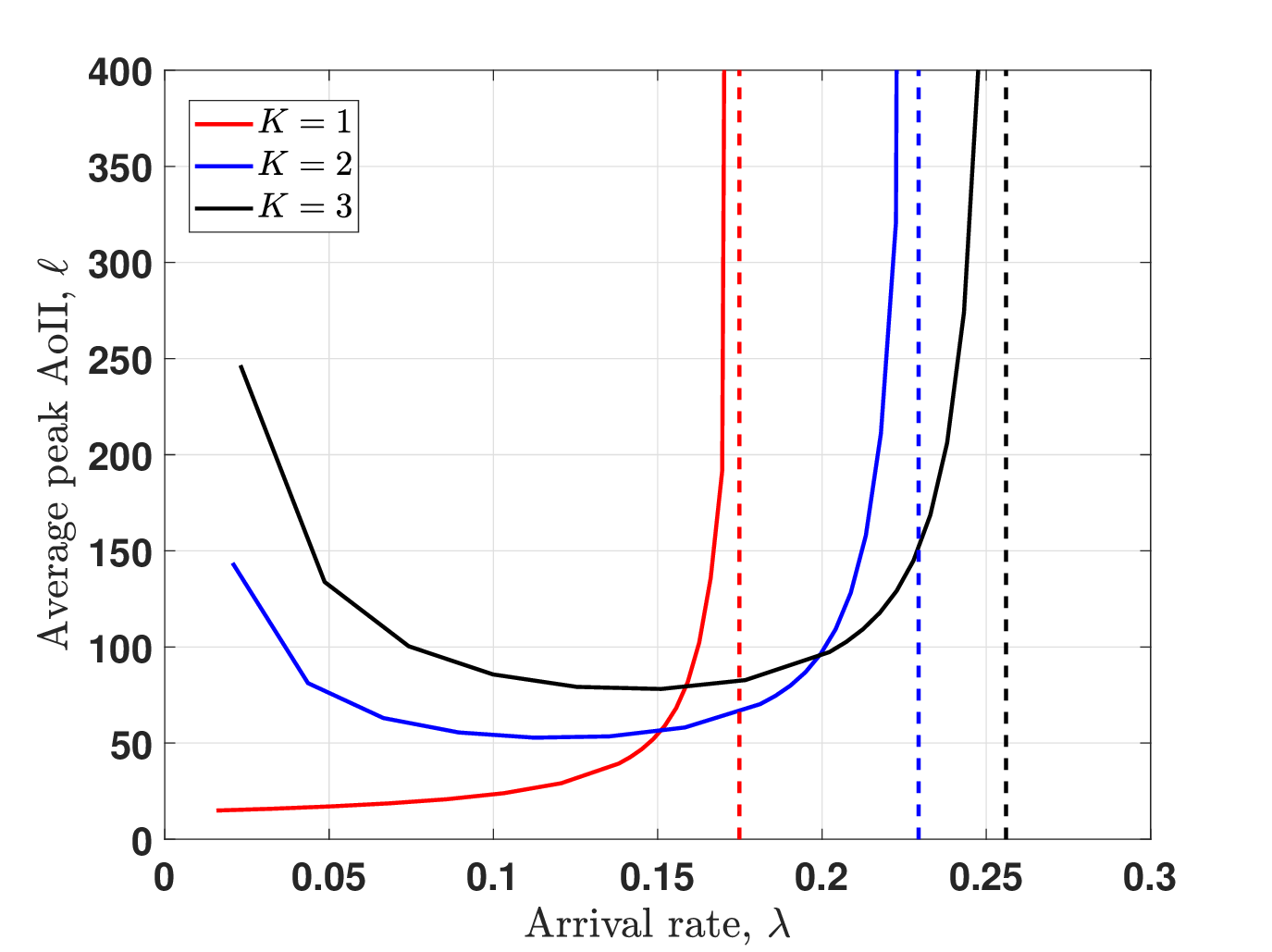}\\
		\caption{The average peak AoII $\ell$ versus the arrival rate $\lambda$ under FD scheme.}
		\label{fig_FD_delay_vs_arrival_rate}
	\end{figure}
	
	Next, we turn our attention to the XD scheme for average peak AoII oriented random access. We obtain that under the optimal allocation scheme, all bandwidth are allocated to either the reservation signal or the data packets in each timeslot. The average peak AoII $\ell$ versus the arrival rate $\lambda$ with different constraint of the transmission queue $N_2$ is shown in Fig. \ref{fig_ell_versus_arrival_rate}. The number of users is $M = 5$. When the arrival rate is low, the transmission queue can hardly contains much packets. Thus, the average peak AoII with different $N_2$ are nearly the same with each other and increase slowly with the arrival rate. As the arrival rate increases, more packets may arrive resulting in more collisions for the reservation. Thus, the average peak AoII $\ell$ increases rapidly with the arrival rate. In this case, larger transmission queue length constraint $N_2$ may allow more nodes to make reservations when the number of nodes in the reservation queue is small, which alleviates the collision hence significantly reducing the average peak AoII.
	
	The average peak AoII $\ell$ versus the maximum queue length $N_2$ with different arrival rate $\lambda$ is shown in Fig. \ref{fig_ell_versus_queue_length}. The arrival rate is given by $\lambda = 0.15$, $\lambda = 0.2$, $\lambda = 0.25$, and $\lambda = 0.3$. The number of users is $M = 10$. If the maximum queue length $N_2$ of the transmission queue is too small, some nodes have to wait in the reservation queue until the nodes in the transmission queue finishes transmission when they can make reservations. Thus, nodes can be congested in the reservation queue, which results in more collisions and higher peak AoII. When $N_2 = 0$, it is similar to the CSMA/CA scheme in which a node begins transmission right after it successfully sends a reservation signal. Since the average peak AoII decreases with $N_2$, the presented dynamic allocation scheme outperforms the CSMA/CA protocol when the number of nodes in the reservation queue can be obtained through collision level estimation. When the arrival rate is small, the average peak AoII decreases to the minimum with any $N_2$. When the arrival rate is large, the average peak AoII is further reduced with a larger $N_2$.
	
	\begin{figure}
		\centering
		\includegraphics[width=1\linewidth]{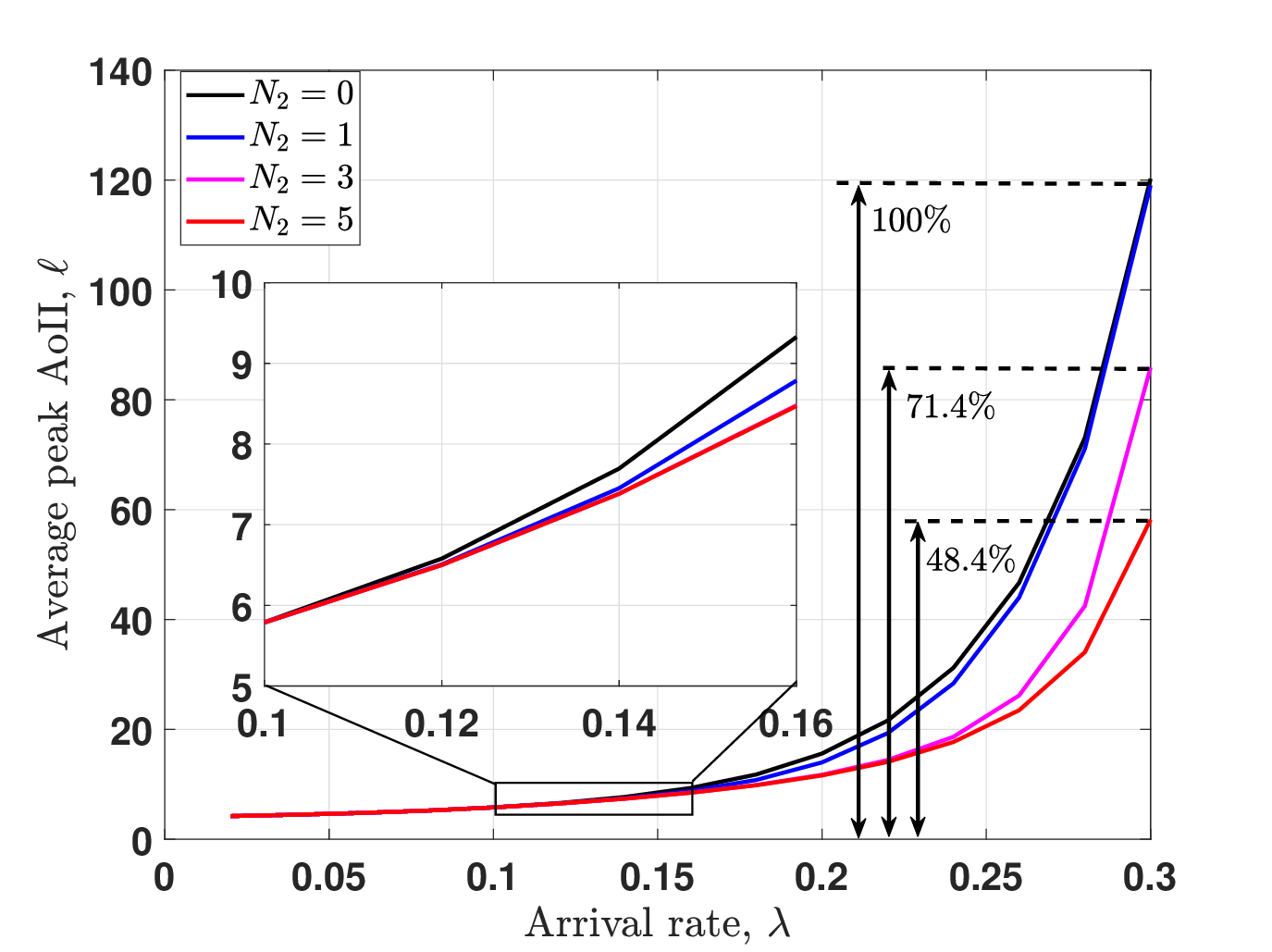}\\
		\caption{The average peak AoII $\ell$ versus the arrival rate $\lambda$ under the XD scheme.}
		\label{fig_ell_versus_arrival_rate}
	\end{figure}
	
	
	\begin{figure}
		\centering
		\includegraphics[width=1\linewidth]{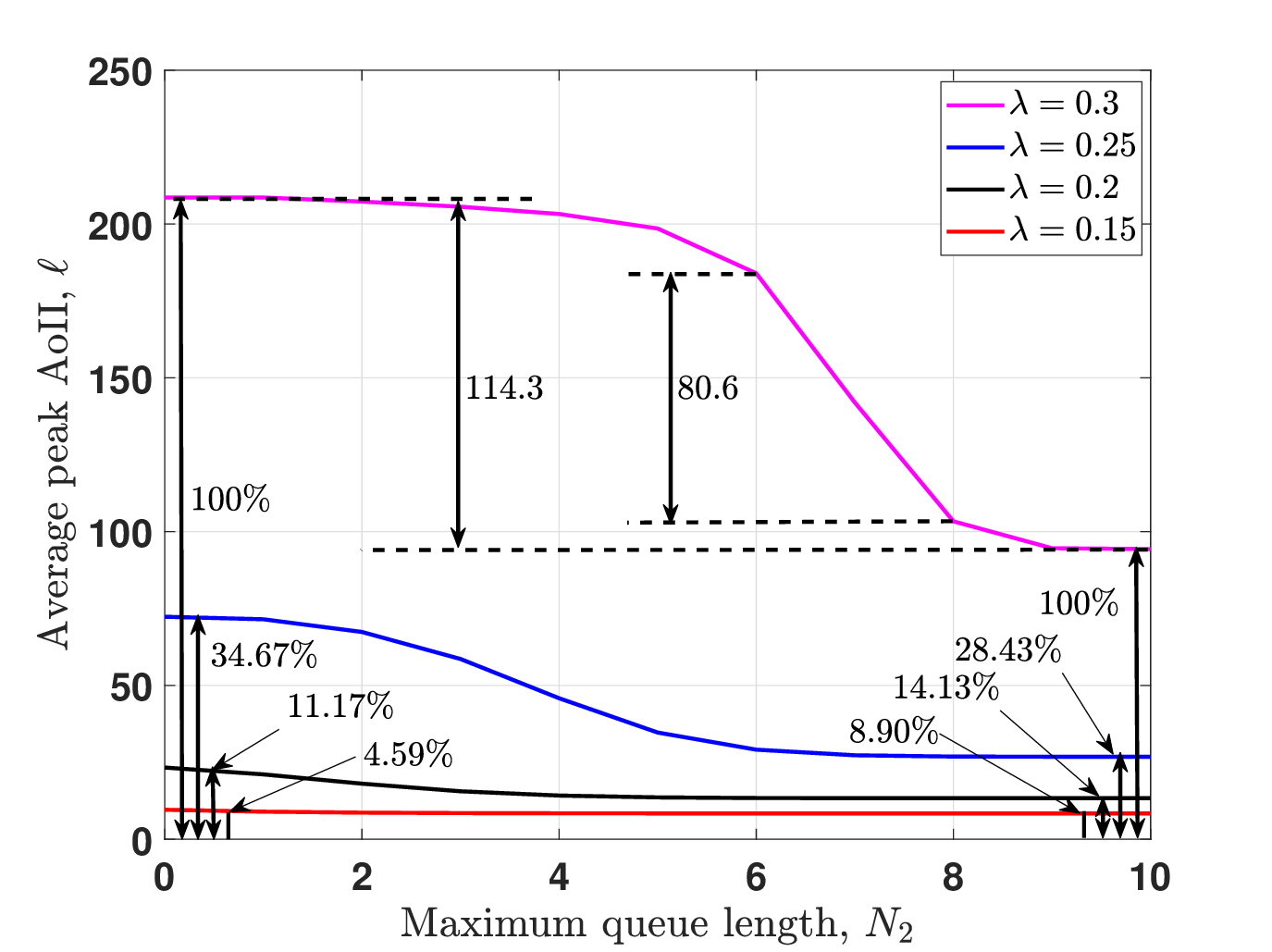}\\
		\caption{The average peak AoII $\ell$ versus the constraint of transmission queue length $N_2$ under the XD scheme.}
		\label{fig_ell_versus_queue_length}
	\end{figure}
	
	The average peak AoII $\ell$ versus the number of nodes $M$ in the network with different arrival rate is shown in Fig. \ref{fig_ell_versus_num_nodes}. The arrival rate is given by $\lambda = 0.1$, $\lambda = 0.15$, and $\lambda = 0.2$. We consider two typical cases of transmission queue constraints including $N_2 = 0$ and $N_2 = M$. It is shown that the average peak AoII does not increase with the number of nodes only when the total arrival rate is low. With a larger arrival rate, the average peak AoII with $N_2=0$ increases more rapidly, while the average peak AoII without constraint of the transmission queue increases slowly and approaches a constant as the number of nodes increases. This indicates that the dynamic allocation scheme is scalable with the number of nodes by significantly alleviating the collision among nodes. Thus, it holds the promise of addressing the access control for a massive number of nodes.

	\begin{figure}
		\centering
		\includegraphics[width=1\linewidth]{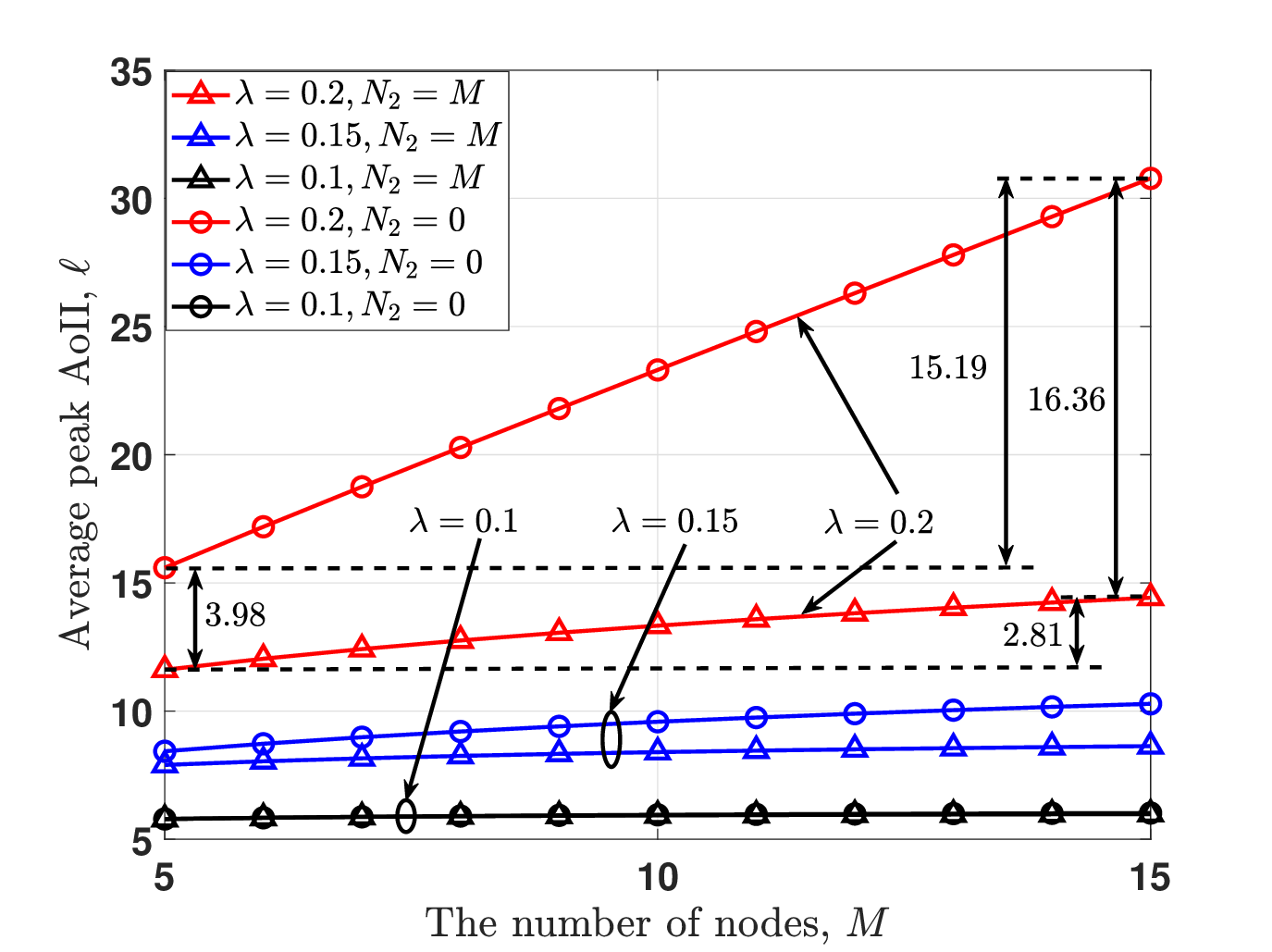}\\
		\caption{The average peak AoII $\ell$ versus the number of transmitter nodes $M$ under the XD scheme.}
		\label{fig_ell_versus_num_nodes}
	\end{figure}
	
	Next, we compare the TD, FD, and XD schemes in Fig. \ref{fig_comparison}. The number of nodes is $M = 5$. The XD scheme significantly outperforms the TD and FD schemes since the dynamic allocation scheme can avoid much waste of the bandwidth resources. Moreover, the XD scheme can achieve higher network throughput with finite average peak AoII. In The FD scheme, the access trigger stage with a higher waiting number $K$ improves the achievable network throughput with the cost of extra peak AoII. The adaptive reservation and transmission scheme with $K = 1$ and $N_{\mathrm{max}}=\infty$ significantly reduces the average peak AoII since the data packet does not wait for more packets before the reservation. Releasing the constraints of the transmission queue in the XD scheme can further reduce the average peak AoII since nodes can make reservations with less collision in the reservation queue while waiting more time in the transmission queue.
	
	
	\begin{figure}
		\centering
		\includegraphics[width=1\linewidth]{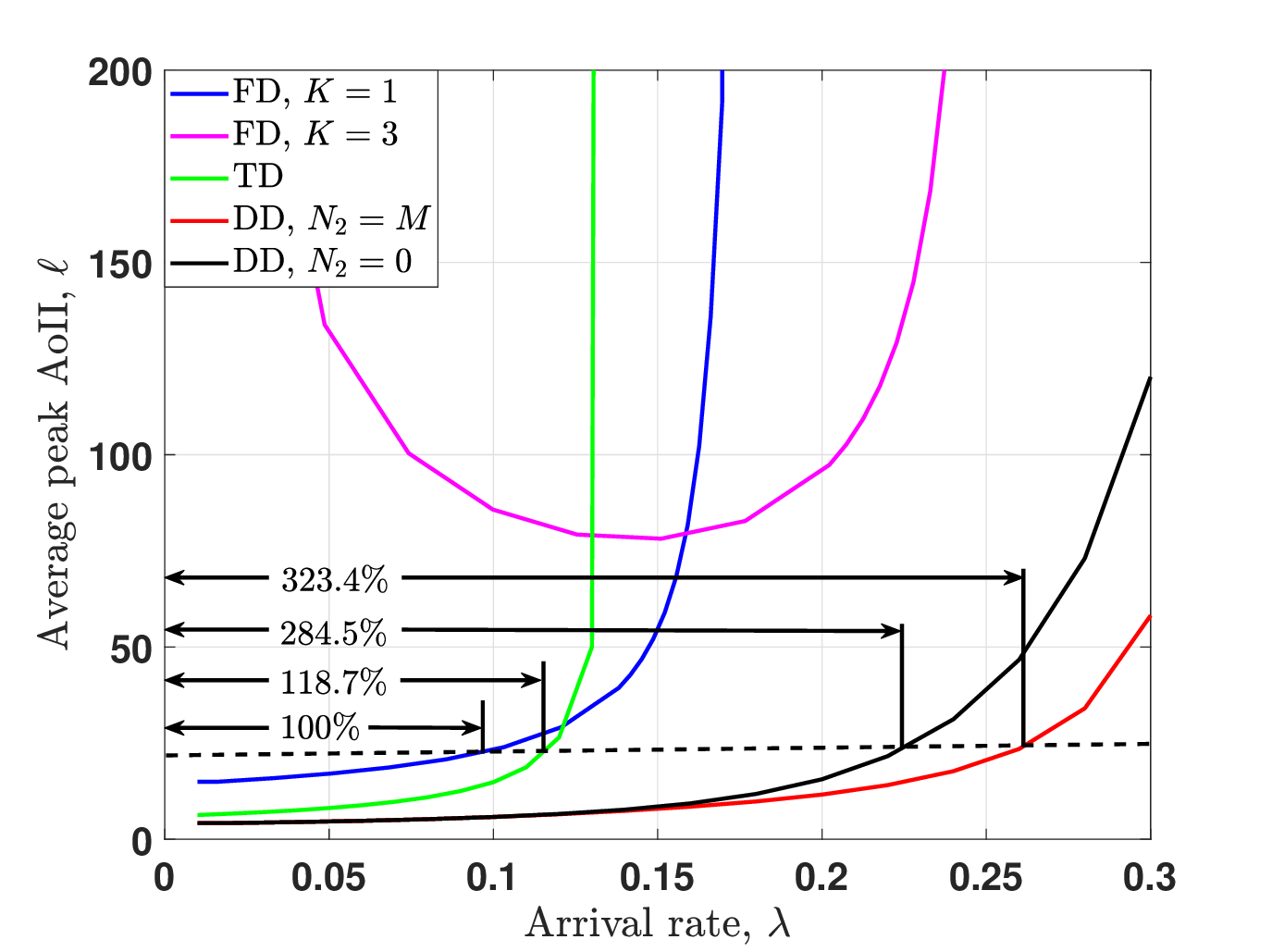}\\
		\caption{Comparison between the TD, FD, and XD schemes. }
		\label{fig_comparison}
	\end{figure}
	
	Finally, we demonstrate the effectiveness of the mean-field  approximations average AoII and peak AoII with a large amount of transmitter nodes. The average AoII $\bar{\Delta}$ versus the number of transmitter nodes $M$ is shown in Fig. \ref{fig_meanfield_aoii}. The total arrival rate of all transmitter nodes is $\lambda = 0.2$ and $\lambda=0.4$. It is shown that the average AoII increases with the number of transmitter nodes $M$. The mean-field approximation approaches the simulation results when the number of transmitter nodes becomes large enough. When the arrival rate is relatively large, the mean-field  approximation can be quite close even with $M=20$. In addition, the average peak AoII $\ell$ versus the number of transmitter nodes $M$ is shown in Fig. \ref{fig_meanfield_delay}. The total arrival rate is $\lambda=0.2$ and $\lambda = 0.25$. As the average peak AoII increases with the number of transmitter nodes, the accuracy of the mean-field  approximation increases. Moreover, the mean-field  approximation is more accurate with larger arrival rate. Thus, the mean-field approximations is also applicable for scenarios with heavy traffic or massive number of transmitter nodes. 
	
	\begin{figure}
		\centering
		\includegraphics[width=1\linewidth]{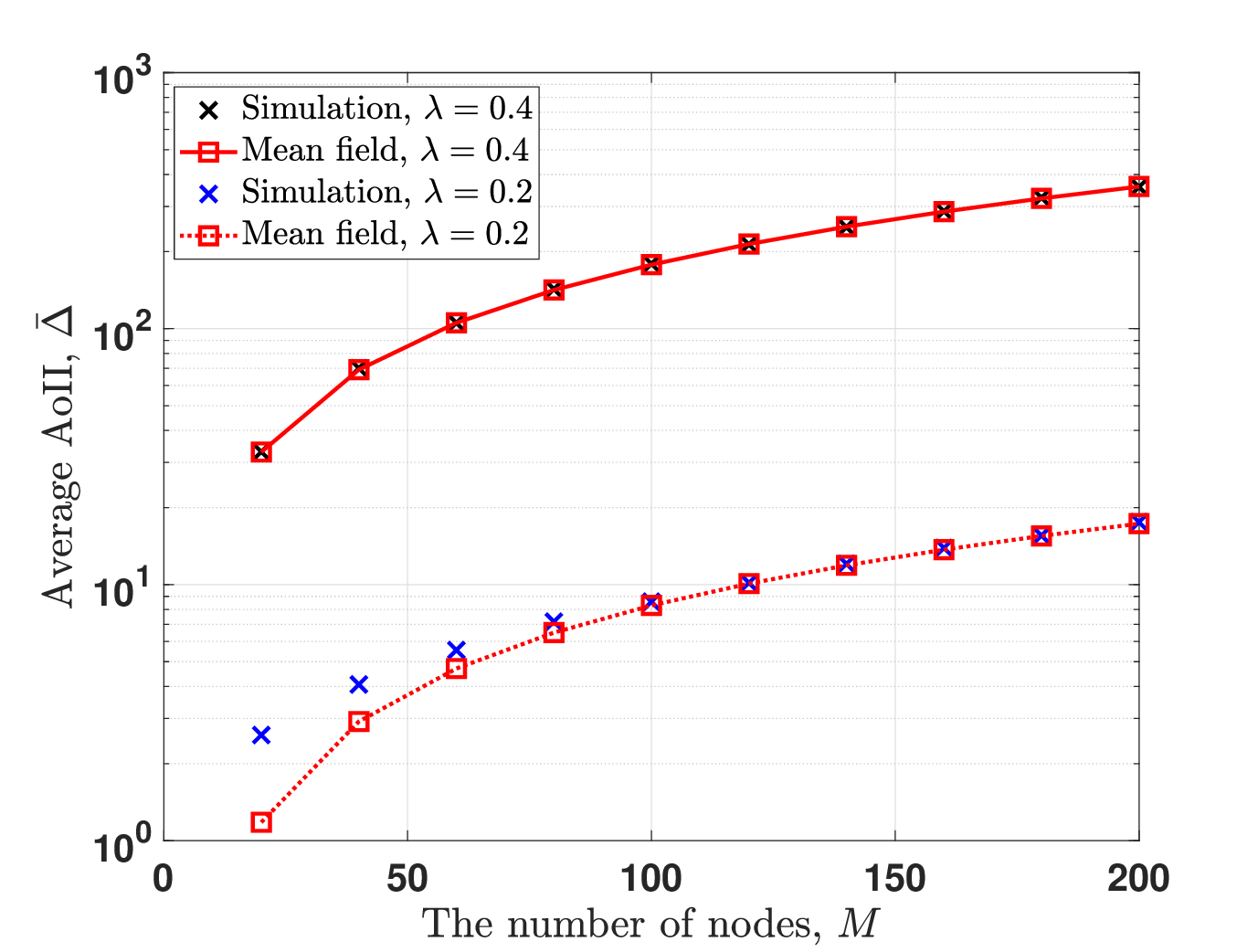}\\
		\caption{The average AoII $\bar{\Delta}$ versus the number of transmitter nodes $M$. }
		\label{fig_meanfield_aoii}
	\end{figure}
	\begin{figure}
		\centering
		\includegraphics[width=1\linewidth]{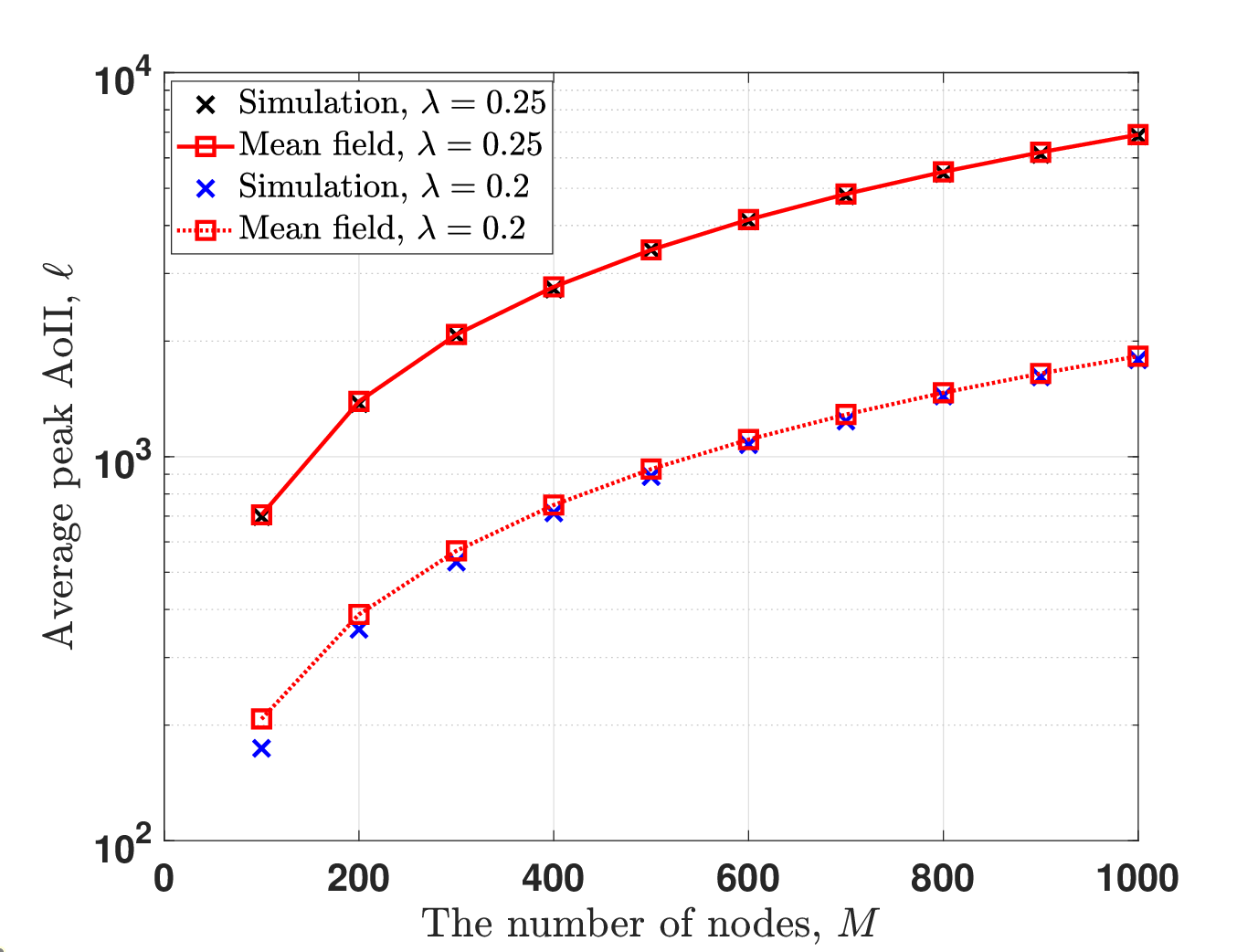}\\
		\caption{The average peak AoII $\ell$ versus the number of transmitter nodes $M$. }
		\label{fig_meanfield_delay}
	\end{figure}
	
	\section{Conclusion}

	In this paper, we have built a unified framework based on large models and mean-field approximations. Based on the unified framework, we analyze freshness-oriented multiple access, which is referred to as fresh multiple access, focusing on AoII and peak AoII scenarios. The average AoII and average peak AoII are analyzed for multiple access schemes characterized by Markov chain model. For AoII-oriented cases, we formulate a large Markov model for general multiple access schemes using our unified framework. We devise an algorithm to derive a sparse transition matrix for this model, enabling efficient computation of large, high-dimensional Markov models. In scenarios centered on average peak AoII, we reduce dimensionality through integral states of the system.  Three multiplexing schemes of reservation signals and data packets are studied for the analysis and optimization of average peak AoII. Moreover, to address the high dimensional model with massive users, mean-field approximations is presented to approximate the impact of all nodes by a simple statistical effect. Using the mean-field approximation, we analyze the average AoII and average peak AoII based on a small Markov model of a single node. Extensive simulations are presented to demonstrate the analysis for AoII and average peak AoII based on different Markov chain model formulation. However, it's important to note that we currently do not address the class of random access schemes that depend on the AoI or AoII. This area represents a significant aspect of our future research focus.
	
	
	\appendices
	\section{TD Multiplexing Scheme for Peak-AoII Oriented Scenario}
	\label{appendix_TD}
	
	
	For simplicity, we consider the maximum length of the reservation queue and the transmission queue are the same, given by $N = N_1 = N_2$. We present the Markov chain for the TD scheme with any collision resolution algorithm defined by $X_0$, $X_1$, and $Y_n$. Let $t_d \in \{1,2,\ldots, Z_1+cZ_2\}$ denote the index of the timeslot within a frame. Let $q_1$ and $q_2$ denote the length of the reservation queue and the transmission queue, respectively. Let $s$ denote the state of the collision resolution algorithm. The state of the tandem queue is given by $(t_d,q_2,q_1,s)$. The state transition of $s$ is defined by $X_0$, $X_1$, and $Y_n$. In the following, we present the state transition for states $(q_1,s)$, $(q_2,q_1,s)$, and finally for states $(t_d,q_2,q_1,s)$.
	
	The length of the reservation queue $q_1$ depends on the arrival of packets and the transmission of reservation signals. When a reservation signal is successfully transmitted in the current timeslot, the length of the reservation queue changes from $q_1$ to $q_1'$ with $q_1'-q_1+1$ new packets arrivals at nodes. The transition matrix for states $(q_1, s)$ in the current timeslot is defined as a $(N+1)r_0 \times (N+1)r_0$ matrix. In the first $Z_1$ timeslots of a frame, the transition matrix for state $s$ when $q_1 = i$ is given by $Y_i X_1$. Thus, the transition matrix for state $(q_1, s)$ in the first $K$ timeslots of a frame is given by
	\begin{align}
		B = \left[\begin{array}{ccccc}
			0 & 0 & \cdots & 0 & 0 \\
			B_{10} & B_{11} & \cdots & B_{1(N-1)} & B_{1N} \\
			0 & B_{21} & \cdots & B_{2(N-1)} & B_{2N} \\
			\vdots & \ddots & \ddots & \vdots & \vdots \\
			0 & \cdots & 0 & B_{N(N-1)} & B_{NN}
		\end{array}\right],
	\end{align}
	in which $B_{ij}$ is an $r_0 \times r_0$ matrix characterizing the transition of state $s$ when the state $q_1$ transits from $q_1 = i$ to $q_1' = j$. For $j = 0, \ldots, N-1$, the matrix $B_{ij}$ is given by
	\begin{align}
		B_{ij} = a_{j-i+1} Y_i X_1.
	\end{align}
	If more than $N$ packets arrive at nodes, then only $N$ packets can be left in the reservation queue due to the constraint while other packets should be dropped. Thus for $j = N$, the matrix $B_{iN}$ is given by
	\begin{align}
		B_{iN} = \sum_{k=N-i+1}^{\infty} a_k Y_i X_1.
	\end{align}
	
	Similarly, we present the state transition for stats $(q_1, s)$ when no reservation signal is successfully transmitted in the current timeslot. The transition matrix for states $(q_1, s)$ in this case is given by
	\begin{align}
		A_0 = \left[\begin{array}{cccc}
			A_{00} & A_{01} & \cdots & A_{0N} \\
			0      & A_{11} & \cdots & A_{1N} \\
			\vdots & \ddots & \ddots & \vdots \\
			0      & \cdots & 0      & A_{NN}
		\end{array}\right],
	\end{align}
	in which $A_{ij}$ is an $r_0 \times r_0$ matrix characterizing the transition of state $s$ when the state $q_1$ transits from $q_1 = i$ to $q_1' = j$. The matrix $A_{ij}$ is given by
	\begin{align}
		A_{0,ij} = \begin{cases}
			a_{j-i} Y_i X_0, & 0 \leq j < N, \\
			\sum_{k=j-i}^{\infty}a_k Y_i X_0, & j = N.
		\end{cases}
	\end{align}
	
	In the last $cZ_2$ timeslots of a frame, the state of the collision resolution $s$ does not change since no bandwidth is allocated to reservation signals. Thus, the transition matrix for state $s$ is given by an identity matrix $I_{r_0}$ whose diagonal elements are equal to one while other elements are equal to zero. We denote the transition matrix for states $(q_1, s)$ in the last $cX$ timeslots of a frame by $A_1$ with elements given by
	\begin{align}
		A_{1,ij}' = \begin{cases}
			a_{j-i} I_{r_0}, & 0 \leq j < N, \\
			\sum_{k=j-i}^{\infty}a_k I_{r_0}, & j = N.
		\end{cases}
	\end{align}
	
	The length of the transmission queue $q_2$ depends on successful transmission of reservation signals and the transmission of data packets. The transition matrix for states $(q_2, q_1, s)$ is defined as a $(N+1)^2r_0 \times (N+1)^2r_0$ matrix. Specifically, in the first $K$ timeslots of a frame, the state $q_2$ may transit from $q_2 = i$ to $q_2' = i+1$ if a reservation signal is successfully transmitted or to $q_2' = i$ otherwise. Thus, the transition matrix for states $(q_2, q_1, s)$ in the first $Z_1$ timeslots of a frame is given by
	\begin{align}
		C_1 = \left[\begin{array}{ccccc}
			A_0 & B & 0 & \cdots & 0 \\
			0 & A_0 & B & \ddots & \vdots \\
			\vdots & \ddots & \ddots & \ddots & 0 \\
			0 & \cdots & 0 & A_0 & B \\
			0 & \cdots & 0 & 0 & A_1
		\end{array}\right].
	\end{align}
	Note that no more reservation signal is allowed to be transmitted when the length of the transmission queue is $q_2 = N$ due to the constraint of queue length. Thus, when $q_2 = N$, the transition matrix for states $(q_1, s)$ is $A_1$. For simplicity, we define four auxiliary matrices given by
	\begin{align}
		L_1 = \left[\begin{array}{cc}
			I_N & \bm{0}_N \\
			\bm{0}_N^{\mathrm{T}} & 0
		\end{array}\right],
		L_2 = \left[\begin{array}{cc}
			\bm{0}_N & I_N \\
			0 & \bm{0}_N^{\mathrm{T}}
		\end{array}\right], \nonumber\\
		L_3 = \left[\begin{array}{cc}
			\bm{0}_{N \times N} & \bm{0}_{N} \\
			\bm{0}_{N}^{\mathrm{T}} & 1
		\end{array}\right],
		L_4 = \left[\begin{array}{cc}
			1 & \bm{0}_N^{\mathrm{T}} \\
			I_N & \bm{0}_N \\
		\end{array}\right],
	\end{align}
	in which $\bm{0}_N$ is a $N \times 1$ column vector with all elements equal to zero, and $\bm{0}_{N \times N}$ is a $N \times N$ matrix with all elements equal to zero. Using Kronecker product, the matrix $C_1$ is represented by $A_0$, $A_1$, and $B$, given by
	\begin{align}
		C_1 = L_1 \otimes A_0 + L_2 \otimes B + L_3 \otimes A_1.
	\end{align}
	Moreover, in the last $cZ_2$ timeslots of a frame, the length of transmission queue $q_2$ decreases by one for every $Z_2$ timeslots after a packet is transmitted. Thus, in timeslots $t_d=Z_1+iZ_2$ of a frame for $i = 1, \ldots, c$, the transition matrix is given by
	\begin{align}
		C_2 = I_{N+1} \otimes A_1.
	\end{align}
	In timeslots $t_d \neq Z_1+iZ_2$ of a frame, the transition matrix is given by
	\begin{align}
		C_3 = L_4 \otimes A_1.
	\end{align}
	
	Finally, we present the transition matrix for states $(t_d, q_2, q_1, s)$. The index of timeslot $t_d$ increases by one for each timeslot. Thus, the transition matrix for states $(t_d, q_2, q_1, s)$ is given by
	\begin{align}
		D = \left[\begin{array}{ccccc}
			0 & D_{1} & 0 & \cdots & 0 \\
			0 & 0 & D_{2} & \ddots & \vdots \\
			\vdots & \vdots & \ddots & \ddots & 0 \\
			0 & 0 & \cdots & 0 & D_{K+cT-1} \\
			D_{K+cT} & 0 & \cdots & 0 & 0
		\end{array}\right],
	\end{align}
	in which $D_i$ is the transition matrix of states $(q_2, q_1, s)$ in the $i$th timeslot of a frame for $i = 1, \ldots, Z_1+cZ_2$. Therefore, The transition matrix $D$ is of order $r_2 = (Z_1+cZ_2)(N+1)^2r_0$ with elements given by
	\begin{align}
		D_i = \begin{cases}
			C_1, & 1 \leq i \leq Z_1, \\
			C_2, & i = Z_1+jZ_2~\mathrm{for}~1\leq j \leq c, \\
			C_3, & Z_1\!+\!(j\!-\!1)Z_2 < i < Z_1\!+\!jZ_2~\mathrm{for}~ 1 \leq j \leq c.
		\end{cases}
	\end{align}
	
	Based on the Markov chain for the TD scheme, we can obtain the steady-state probabilities of all states. Let a row vector $\bm{\pi}$ denote the steady-state probabilities of all states. The steady-state probabilities satisfy that
	\begin{align}
		\label{steady_eq_1}
		\begin{cases}
			\bm{\pi} D = \bm{\pi}, \\
			\bm{\pi} \bm{1}_{r_2} = 1,
		\end{cases}
	\end{align}
	in which $\bm{1}_{r_2}$ represents a $r_2 \times 1$ columns vector with all elements equal to one. From Eq. (\ref{steady_eq_1}) we can obtain that the steady-state probability $\bm{\pi}$ satisfy that $\bm{\pi}(D - I_{r_2} + \bm{1}_{r_2 \times r_2}) = \bm{1}_{r_2}^{\mathrm{T}}$, in which $\bm{1}_{r_2 \times r_2}$ represents an $r_2 \times r_2$ matrix with all elements equal to one. Then we can obtain the steady-state probabilities $\bm{\pi}$ by
	\begin{align}
		\bm{\pi} = \bm{1}_{r_2}^{\mathrm{T}} \left(D - I_{r_2} + \bm{1}_{r_2\times r_2}\right)^{-1}. 
	\end{align}
	For high-dimensional Markov chains with a large transition matrix, the matrix can be quite sparse according to the definition. Thus, advanced methods for solving large and sparse Markov chains can be adopted to significantly improve the computational efficiency \cite{large_sparse_chain_1, large_sparse_chain_2}. 
	
	The steady-state probability for state $(t_d, q_2, q_1, s)$ is denoted by $\pi_{t_d,q_2,q_1,s}$. Using the steady-state probabilities, we can obtain the average length of the tandem queue given by
	\begin{align}
		\bar{L} = \sum_{q_1 = 0}^{N} \sum_{q_2 = 0}^{N} (q_1 + q_2) \sum_{t_d=1}^{Z_1+cZ_2} \sum_{s \in \mathcal{S}_{\mathrm{R}}} \pi_{t_d,q_2,q_1,s}.
	\end{align}
	The average peak AoII is obtained according to Eq. (\ref{eq_average_delay}). However, the actual arrival rate to the reservation queue may not be equal to the packet arrival rate $\lambda$ due to the constraint of queue length $N$. We next derive the actual arrival rate to the reservation queue $\lambda'$. Denote the probability by $\pi_i'$ if there are $i$ packets in the reservation queue after the transmission of reservation signal in a timeslot. Before the transmission of reservation signal in that timeslot, the length of the reservation queue is $i+1$ or $i$. For given $t_d$, $q_2$, and $q_1$, denote the steady-state probabilities for $s \in \mathcal{S}_{\mathrm{R}}$ by a $1 \times r_0$ row vector $\bm{\pi}_{t_d,q_2,q_1} = \left(\pi_{t_d,q_2,q_1,s_1}, \ldots, \pi_{t_d,q_2,q_1,s_{r_0}}\right)$. The probability $\pi_i'$ is given by
	\begin{align}
		\pi_i' =& \sum_{t_d=1}^{Z_1} \left[\sum_{q_2=0}^{N-1} \left(\bm{\pi}_{t_d,q_2,i} X_0 + \bm{\pi}_{t_d,q_2,i+1}  X_1 \right)\bm{1}_{r_0} + \bm{\pi}_{t_d,N,i}\right] \nonumber\\
		&+ \sum_{t_d=Z_1+1}^{Z_1+cZ_2} \sum_{q_2=0}^{N}\bm{\pi}_{t_d,q_2,i}.
	\end{align}
	When there are $i$ packets in the reservation queue, at most $N-i$ packets can arrive at the reservation queue at the end of the timeslot. Thus, the actual arrival rate $\lambda'$ is given by
	\begin{align}
		\lambda' = \sum_{i=0}^{N}\left(\sum_{j=0}^{N-i-1} j a_j \pi_i' + (N-i)\sum_{j=N-i}^{\infty}a_j \pi_i' \right).
	\end{align}
	According to Eq. (\ref{eq_average_delay}), the average peak AoII is obtained by $\ell = \frac{\bar{L}}{\lambda'}$. Moreover, the packet loss rate is given by $1 - \frac{\lambda'}{\lambda}$.

	\begin{spacing}{0.91}
	
\end{spacing}
	
	\begin{IEEEbiography}[{\includegraphics[width=1in,height=1.25in,clip,keepaspectratio]{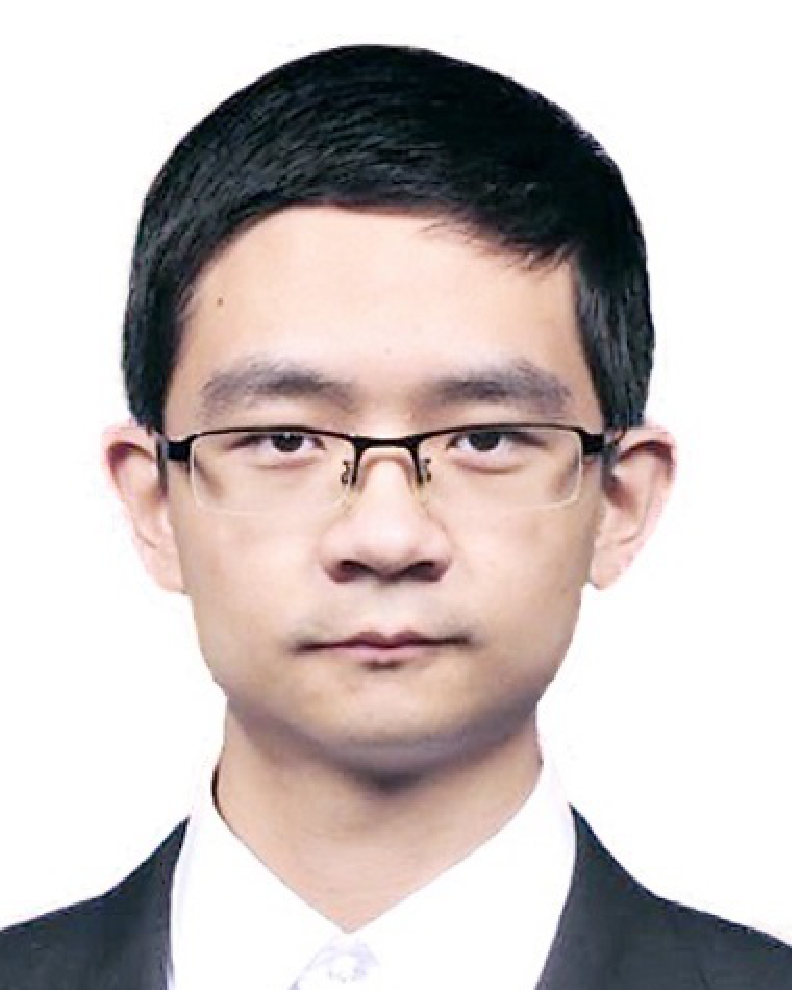}}]{Haiming Hui}
		received the B.S. degree from Tsinghua University, Beijing, China, in 2018, where he is currently pursuing the Ph.D. degree with the Department of Electronic Engineering. His research interest includes age of information and the proactive pushing techniques in wireless networks.
	\end{IEEEbiography}
	
	\vfill
	
	\begin{IEEEbiography}[{\includegraphics[width=1in,height=1.25in,clip,keepaspectratio]{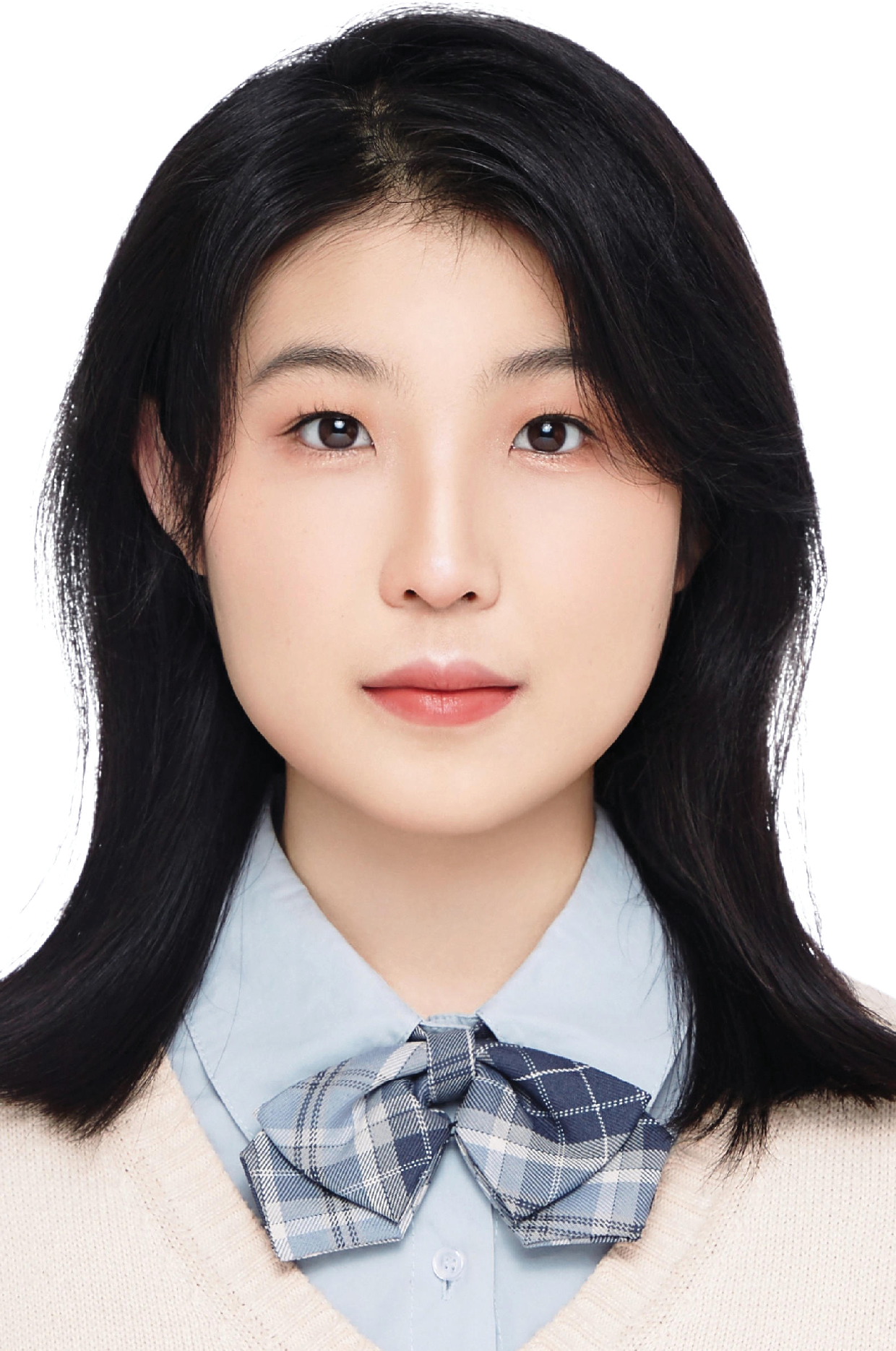}}]{Shuqi Wei}
		received the B.S. degree from Tsinghua University, Beijing, China, in 2023, where she is pursuing the Ph.D. degree with the Department of Electronic Engineering. Her research interests are in the areas of real-time communications and task-oriented communications.
	\end{IEEEbiography}
	\vfill
	
	\begin{IEEEbiography}[{\includegraphics[width=1in,height=1.25in,clip,keepaspectratio]{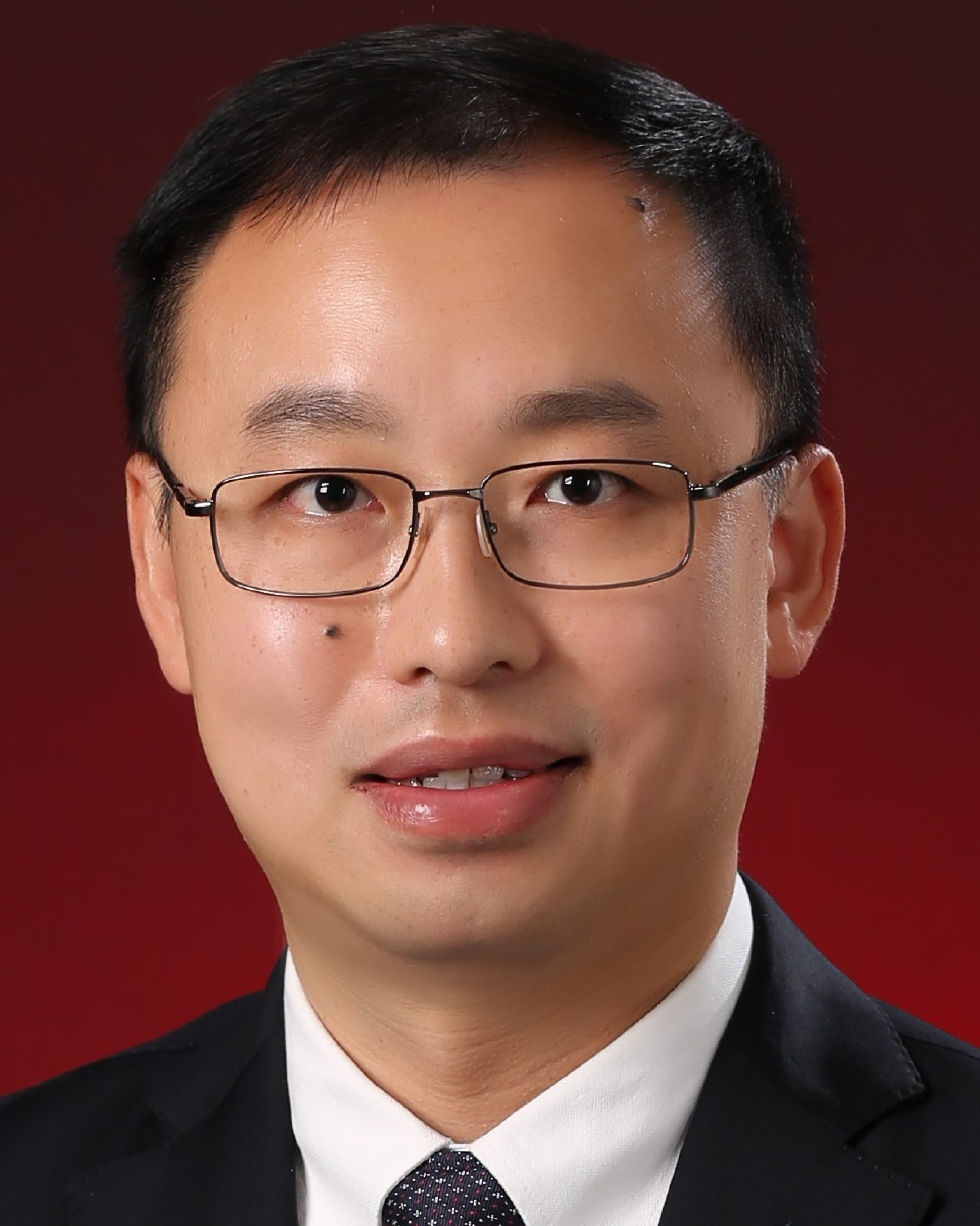}}]{Wei Chen}
		(Senior Member, IEEE) received the B.S. and Ph.D. degrees (Hons.) from Tsinghua University in 2002 and 2007, respectively. Since 2007, he has been a Faculty Member with Tsinghua University, where he is currently a Tenured Full Professor and a University Council Member. During 2016-2021, he was the Director of the Degree Office of Tsinghua University. During 2014-2016, he was the Deputy Head of the Department of Electronic Engineering in Tsinghua University. From 2005 to 2007, he was a Visiting Ph.D. Student with the Hong Kong University of Science and Technology. He visited the University of Southampton in 2010, Telecom Paris Tech in 2014, and Princeton University, Princeton, NJ, USA, in 2016. His research interests are in the areas of real-time communications, task-oriented communications, and co-design of communications, control, and optimizations.
		
		He is a Cheung Kong Young Scholar and a member of the National Program for Special Support of Eminent Professionals, also known as 10,000 talent program. He received the IEEE Marconi Prize Paper Award in 2009 and the IEEE Comsoc Asia Pacific Board Best Young Researcher Award in 2011. He is a recipient of the National May 1st Labor Medal and the China Youth May 4th Medal. He has also been supported by the National 973 Youth Project, the NSFC Excellent Young Investigator Project, the New Century Talent Program of the Ministry of Education, and the Beijing Nova Program. He serves as an Editor for IEEE TRANSACTIONS ON WIRELESS COMMUNICATIONS. He also serves as a standing committee member of All-China Youth Federation and the secretary-general of its education board. He has served as Editors for IEEE TRANSACTIONS ON COMMUNICATIONS and IEEE WIRELESS COMMUNICATIONS LETTERS, a TPC Co-Chair for IEEE VTC-Spring in 2011 and a Symposium Co-Chair for IEEE ICC and Globecom.

	\end{IEEEbiography}
	\vfill
	
\end{document}